\documentclass[12pt]{amsart}

\usepackage{amsmath, amsfonts, amssymb,mathtools,booktabs}\usepackage{dsfont}

\usepackage[foot]{amsaddr}

\usepackage[colorlinks=true, pdfstartview=FitV, linkcolor=blue,citecolor=blue, urlcolor=blue]{hyperref}
\usepackage{cleveref}
\usepackage[shortlabels]{enumitem}
\usepackage{tikz}

\usepackage[margin=1in]{geometry}
\usepackage{setspace}
\AtBeginDocument{  \setlength{\abovedisplayskip}{5pt plus 1pt minus 1pt}
  \setlength{\belowdisplayskip}{5pt plus 1pt minus 1pt}
  \setlength{\abovedisplayshortskip}{2pt plus 1pt minus 1pt}
  \setlength{\belowdisplayshortskip}{3pt plus 1pt minus 1pt}
}

\makeatletter
\def\paragraph{\@startsection{paragraph}{4}  \z@\z@{-\fontdimen2\font}  {\normalfont\bfseries}}
\makeatother

\newcommand*\diff{\mathop{}\!\mathrm{d}}

\def\CV{\text{CV}}

\def\cvh{\mathrm{cvh}}

\renewcommand{\subset}{\subseteq}
\def\ba{\mathrm{ba}}

\def\dimU{n}
\def\prizes{\mathcal{P}}

\newlength{\twinwd}

\def\one{\mathbf{1}}

\def\B{\mathcal{B}}
\def\D{\mathcal{D}}
\def\K{\mathcal{K}} 

\def\H{\mathcal{H}}

\def\F{\mathcal{F}}

\def\Re{\mathbf{R}} \def\R{\mathbf{R}}

\def\E{\mathbf{E}}

\def\ep{\varepsilon}
 
\def\ta{\theta}
\def\al{\alpha}
\def\la{\lambda}
\def\da{\delta}
 
\def\phi{\varphi}
\def\sa{\sigma}

\def\one{\mathbf{1}}

\def\os{\emptyset}

\newcommand{\df}[1]{\textit{#1}}

\newcommand{\norm}[1]{\| #1 \|}
\newcommand{\ip}[1]{\langle #1 \rangle}
\newcommand{\abs}[1]{ \left | #1 \right | }

\def\nagents{m}

\def\smallskip{\vspace{.3cm}}
\def\medskip{\vspace{.8cm}}
\def\bigskip{\vspace{1.5cm}}

\newcommand{\argmax}{\operatornamewithlimits{argmax}}

\newtheorem{proposition}{Proposition}
\newtheorem{corollary}{Corollary}
\newtheorem{theorem}{Theorem}
\newtheorem{lemma}{Lemma}
\theoremstyle{remark}
\newtheorem{remark}{Remark}
\newtheorem*{example}{Example}

\theoremstyle{definition}
\newtheorem*{definition}{Definition}

\usepackage{natbib}

\begin{document}

\title[Distributional Welfare Analysis]{Utilitarian Social Choice and Distributional Welfare Analysis}

\author[Echenique]{Federico Echenique}
\author[Valenzuela-Stookey]{Quitz\'{e} Valenzuela-Stookey}
\address{Department of Economics, UC Berkeley}
\email{fede@econ.berkeley.edu,quitze@berkeley.edu}
\thanks{We are deeply indebted to Matt Backus for early discussions on the use of welfare evaluations and discrete-choice models in IO, which inspired our work. Jinglin Yang provided superb research assistance.}

\begin{abstract}    
    The widely-used ``utilitarian'' approach to welfare analysis, based on the axiomatic foundation of \cite{harsanyi1955cardinal}, may overlook important distributional considerations. We therefore introduce a notion of utilitarianism that applies to \textit{social choice functions}, describing the actions of society, rather than social welfare functions describing society's preferences (as in Harsanyi). We characterize a representation of utilitarian social choice, and show that it provides a foundation for a family of \textit{distributional welfare measures} based on quantiles of the distribution of individual welfare effects, rather than averages. We illustrate this approach in the context of the automobile market in \cite{berry1995automobile}. 
\end{abstract}

\begin{titlepage}
\maketitle     \thispagestyle{empty}
\end{titlepage}

 \clearpage\thispagestyle{empty}

 {\small
 \tableofcontents
 }
\setcounter{page}{0}
\clearpage

\section{Introduction}\label{sec:introduction}

How society’s choices should reflect the preferences of its members is a core problem in economics. It is often framed in terms of social preferences: given individuals' preferences over possible outcomes, what preferences over those same outcomes should society aim to maximize? \cite{arrow1950difficulty} famously showed that there is no perfect way to aggregate individual preferences into a collective one without giving up some intuitively desirable properties. Yet, because we still need a concrete standard to judge social choices, a vast literature has emerged to navigate the trade-offs exposed by Arrow’s theorem (see, for instance, \cite{sep-social-choice} for an overview).

For economists, a key reason to study social preferences is to serve as a foundation for welfare analysis. In practice, the dominant approach is some variant of utilitarianism: we specify a utility function that represents each individual's preferences and then evaluate social outcomes using a weighted average of individual utilities. This intuitively appealing criterion is, moreover, supported by seemingly strong axiomatic support. \cite{harsanyi1955cardinal} argues that the utilitarian representation of social preferences is implied by a normatively appealing Pareto principle: if every individual strictly prefers alternative A to alternative B, then society must also prefer A to B.\footnote{Harsanyi avoids Arrow’s impossibility result by restricting preferences to those representable by expected-utility theory and by dropping Arrow’s independence of irrelevant alternatives axiom.}

Nonetheless, there exists a long-running dissatisfaction with utilitarianism as a social-welfare metric (e.g.\ \citealt{rawls1971theory}), and within economics, there is growing interest in welfare analysis that takes inequalities into account \citep{kanbur1994optimal, roemer1998equality, saez2016generalized, card2018firms, chetty2020income, dworczak2021redistribution, alesina2023immigration, backus2024surplus}. This creates a clear tension: Harsanyi's characterization is normatively appealing, but many economists now want to move beyond utilitarianism to embed distributional justice directly into welfare criteria. Resolving this tension is difficult, however, because many proposed alternatives to utilitarianism lack the kind of axiomatic foundations needed to establish their underlying normative principles.

The objectives of this paper are twofold. First, we seek to understand the tension between Harsanyi's Pareto principle and the desire for distributional welfare measures. In particular, we show that a foundation for distributional welfare analysis can be derived by aggregating choices rather than preferences. Focusing on the choice that society makes from a set of alternatives, as opposed to the entire ranking over these alternatives, is one of the main avenues that the literature has explored to escape the conclusion of Arrow's theorem (e.g.\ \citealt{sen_collectivechoice1970, satterthwaite1975strategy, moulin1980on}). \emph{Our key conceptual innovation is to introduce a notion of utilitarianism in this setting by adapting Harsanyi's Pareto principle to apply to social choice functions rather than social preferences.} To do this, we adopt the model of random expected utility of \cite{gul2006random} rather than the deterministic expected utilities assumed by Harsanyi. This allows us to consider weighted averages of random utilities, with the corresponding aggregation of choices into a single choice rule for a utilitarian planner. 

In our main result (Theorem~\ref{thm:main}), we introduce a Pareto criterion tailored to random expected utility and random choice. Our Pareto principle is based on agreements regarding choice probabilities: whenever all agents agree that A should be chosen more often than B, then the planner must also choose A more often than B. More broadly, whenever agents unanimously agree on certain statements regarding the choice probabilities, the planner must also agree. We prove that this probabilistic Pareto criterion holds if and only if the planner behaves according to our version of utilitarianism.

Second, we provide a foundation for a class of \textit{distributional welfare measures} designed for applications where consumers face price changes. These measures are based on quantiles of the distribution of compensating variation in the population. We show that each such measure can be interpreted as the \textit{stochastic} compensating variation corresponding to a utilitarian social choice function of the form characterized in Theorem~\ref{thm:main}. This justifies, for instance, replacing the standard focus on the average compensating variation with the median. Crucially, our framework does not require the social planner to make interpersonal comparisons of preference intensity, thereby overcoming a central limitation of Harsanyi’s theorem as a basis for welfare analysis. From an applied perspective, our results show that considering the median alongside the mean is a principled way to introduce distributional considerations into welfare analysis. 

The paper is organized as follows. In \Cref{sec:motivation} we introduce some of the conceptual difficulties involved in distributional welfare analysis, using as a running example the problem of evaluating changes in prices in a consumer-choice setting. In \Cref{sec:model} we present the general model and notation. In \Cref{sec:collective} we formally present Harsanyi's characterization of utilitarian social preferences, and then introduce our notion of utilitarian social choice. We then characterize the class of utilitarian social choice functions as those which satisfy a suitable Pareto axiom. In \Cref{sec:distributional} we study the implications of the model for applied welfare analysis, including results on the identification of distributional welfare measures from demand data. In \Cref{sec:empirical} we illustrate the application of these results in empirical work by revisiting the welfare analysis of voluntary export restraints \citep{berry1999voluntary} from a distributional perspective. In \Cref{sec:further_discussion} we discuss additional properties of utilitarian social choice. \Cref{sec:alignment} presents an alternative characterization of utilitarian social choice. We elaborate on the connection to the literature in \Cref{sec:literature}.

\section{Motivating application: evaluating price changes}\label{sec:motivation}

Consider a discrete choice setting in which consumers choose among a finite set, $Z$, of goods which differ in their characteristics and prices.\footnote{The subsequent discussion applies to continuous choice as well. The discrete choice assumption only becomes relevant when we discuss estimation in \Cref{sec:estimation}. A similar discussion also applies to other policies, but here we focus on price changes for concreteness.} Suppose we want to answer the question: should society adopt a policy which shifts prices to $p'$ from a baseline of $p^0$? The standard approach is to first use data on prices and market shares to estimate a random utility model of individuals' preferences. That is, each consumer is assumed to solve
\begin{equation*}
 v(p,y|\varepsilon) := \max_{x \in Z} \ u(x, y - p_x |\varepsilon) 
\end{equation*}
 where $y$ is income, $p_x$ is the price of good $x$, and $\varepsilon$ is a vector of preference parameters. For a set of observable consumer types, $I$, the goal is to estimate the family of distributions $\{\hat{F}_i\}_{i\in I}$ such that $\varepsilon \sim \hat{F}_i$ in the population of consumers with type $i \in I$.\footnote{This population-level interpretation is often used interchangeably with the random-utility interpretation, whereby $F_i$ parameterizes a distribution over the preferences of an individual agent $i$.} To simplify the notation, assume that all consumers have the same income (this has no substantive impact on the current discussion). Given the estimated conditional distributions $\{\hat{F}_i\}_{i\in I}$, let $\hat{F}$ be the corresponding marginal distribution over preferences in the population as a whole, i.e. $\hat{F} = \sum_{i\in I} \sigma_i \hat{F}_i$, where $\sigma = (\sigma_i)_{i\in I}$ is the marginal distribution over types. 
 
 With the distribution $\hat{F}$ in hand, answering the policy question of interest entails somehow aggregating individuals' preferences. The early literature on applied welfare analysis identified three main alternatives \citep{hanemann1996welfare}.
\begin{enumerate}
    \item \textbf{\textit{Average indirect utility.}} Define the average indirect utility
    \begin{equation*}
        V(p,y) = \E_{\varepsilon \sim \hat{F}}[v(p,y|\varepsilon)]. 
    \end{equation*}
    Under this approach, we say that the policy shift from $p^0$ to $p'$ should be adopted if $V(p',y) > V(p^0,y)$. This is sometimes referred to as the Bergson-Samuelson approach \citep{bhattacharya2024nonparametric}.  In this setting it is equivalent to quantifying the gains using \textit{aggregate compensating variation}, defined as the value $\CV_{agg}(p')$ satisfying 
    \begin{equation*}
        V(p',y - \CV_{agg}(p')) = V(p^{0},y)
    \end{equation*}
    since $\CV_{agg}(p') \geq 0 \Leftrightarrow V(p',y) \geq V(p^{0},y)$ under the standard assumption that indirect utility is strictly increasing in money. 
    \item \textbf{\textit{Average compensating variation.}} An alternative approach is to first calculate the compensating variation for each preference type: let $\CV(p'|\varepsilon)$ be the value satisfying
    \begin{equation*}
        v(p',y - \CV(p'|\varepsilon)| \varepsilon) = v(p^{0},y|\varepsilon).
    \end{equation*}
    Under this approach, we say that the shift from $p^{0}$ to $p'$ should be adopted if $\CV_{avg}(p') := \E_{\varepsilon \sim \hat{F}}[\CV(p'|\varepsilon)] \geq 0$. 
    \item \textbf{\textit{Median compensating variation.}} Let $\CV_{med}(p')$ be the median $\CV(p'|\varepsilon)$ under distribution $\hat{F}$. Adopt the policy if $\CV_{med}(p') \geq 0$.\footnote{For a distribution with CDF $G$, the median is $\inf\{x : G(x) \geq 0.5\}$. Of course, we can also imagine using other features of the compensating-variation distribution to evaluate the policy, besides the mean and median. We could also replace compensating variation with equivalent variation in all three approaches; the conceptual point remains the same.}
\end{enumerate} 

There are important conceptual differences between these three approaches. The value $\CV_{agg}$ is the answer to the question: under the new prices, what is the maximum payment that could be taken from all consumers so that on average the population is no worse off, using $v$ as a cardinal (interpersonally comparable) measure of individuals' welfare? In contrast $\CV_{avg}$, while superficially similar, answers a very different question: if we could enforce payments conditional on $\varepsilon$, what is the maximum amount of money that could be extracted from the population without making any individual consumer worse off.\footnote{If individuals' utilities are linear in the numeraire, i.e.\ there are no wealth effects, then a consumer's compensating variation is just the difference in their value function, so $\CV_{agg}$ and $\CV_{avg}$ coincide.} Finally, $\CV_{med}$ is the smallest payment that could be taken from all consumers such that at least half of them are no better off. 

The first two measures, $V$ (or $\CV_{agg}$) and $\CV_{avg}$, share the appealing feature that they can be interpreted as the utilitarian aggregation of individuals' preferences over price vectors. That is, both measures induce a social preference, $\succsim$, over alternative price vectors represented by
\begin{equation*}
  p' \succsim p^0 \Leftrightarrow \int w(p'|\varepsilon) \hat{F}(d\varepsilon) \geq \int w(p^0|\varepsilon) \hat{F}(d\varepsilon), 
\end{equation*}
where $w(p'|\varepsilon)$ is a representation of the preferences over price vectors of an $\varepsilon$-consumer. In the case of $V$, $w(p'|\varepsilon)$ is simply the indirect utility function $v(p',y|\varepsilon)$, and for $\CV_{avg}$, $w(p'|\varepsilon) = \CV(p'|\varepsilon)$. For any $\varepsilon$,  $v(p',y|\varepsilon)$ and $\CV(p'|\varepsilon)$ are just different cardinal representations of the same ordinal preferences over price vectors. In a seminal paper, \cite{harsanyi1955cardinal} showed that generalized utilitarian social preferences, in which individual utilities are averaged using \textit{some} distribution, are the only ones which satisfy a normatively appealing \textit{Pareto property}: informally, this means that if all agents rank alternative $x$ above $y$ then so does the social preference. We provide a version of Harsanyi's justification of utilitarianism, stated as Theorem~\ref{thm:zhou} in \Cref{sec:collective}.

Harsanyi's justification provides a conceptual foundation for the welfare measures $V$ and $\CV_{avg}$, in that evaluating the proposed change to $p'$ according to either measure is consistent with the maximization of a normatively appealing social preference. But these foundations are lacking. We shall argue that $V$ and $\CV_{avg}$ cannot be fully justified using Theorem~\ref{thm:zhou}, and that they require additional assumptions that are harder to defend: see Section~\ref{sec:theparetoaxiom}, and the discussion at the end of Section~\ref{sec:harsanyijustification}. 

Each measure, $V$ (or $\CV_{agg}$), $\CV_{avg}$ and $\CV_{med}$, has its own set of limitations, which make them better suited for different tasks.\footnote{\label{footnote:scitovsky} A drawback of $\CV_{avg}$ is that while it is useful for comparing alternative price vectors to a common status-quo $p^0$, the ranking it induces over $p',p'' \neq p^0$ is not necessarily consistent with comparing $p'$ and $p''$ directly: It is possible that $\CV_{avg}(p'') \geq \CV_{avg}(p')$, and yet the average compensating variation when going from $p'$ to $p''$ is negative. In other words, the induced preference over the space of price vectors may contain cycles, a phenomenon sometimes referred to as Scitovsky reversals \citep{bhattacharya2024nonparametric}. This makes $\CV_{avg}$ better suited for evaluating an alternative against the status quo, rather than optimizing in policy space. The aggregate value function $V$ does not suffer from this limitation. On the other hand, $V$ is not invariant to monotone transformations of consumers' utility functions. It therefore cannot be identified directly from the data without some assumptions on interpersonal utility comparisons \citep{bhattacharya2024incorporating}.  In contrast, $\CV_{avg}$ can be calculated directly from Marshallian demand functions, which can be identified non-parametrically in many settings \citep{small1981applied, bhattacharya2015nonparametric, bhattacharya2018empirical}.}  $\CV_{avg}$ has become the default measure in empirical industrial organization, trade, and health-care applications for program evaluation of an intervention against the status quo, whereas variants on the Bergson-Samuelson approach are widely used in public finance for the purposes of characterizing optimal policies (e.g. \citealt{kleven2006marginal, saez2016generalized}). 

In contrast, median compensating variation has, perhaps surprisingly, received much less attention in the literature, despite possessing appealing properties. For one, unlike the mean, evaluations based on the median do not require the planner to take a stand on interpersonal comparisons of cardinal utility. Moreover, as a measure of central tendency for welfare the median has well-known advantages over the mean when the distribution is skewed (e.g. \citealt{kuznets1955economic}). Evaluations based on the mean are sensitive to extreme values, and can mask inequities. These concerns are likely to be particularly salient in the current context, as intuition suggests that the distribution of compensating variation may be skewed, particularly when high-wealth individuals have relatively inelastic demand.\footnote{The observation that median compensation might be a more suitable measure than the mean when the distribution is skewed goes back at least to \cite{hanemann1996welfare}, although previous work does not appear to have articulated the connection between skewness and wealth effects. This latter point is closely related to \cite{backus2024surplus}, who argue that mean compensating variation implies specific regressive weights on agents utility, which favor high-income individuals. Indeed, we are indebted to the authors of \cite{backus2024surplus} for early conversations which motivated our interest in utilitarian social choice.} Indeed, this intuition turns out to be correct, at least for the widely-studied case of Logit demand, where the demand for good $x \in Z$ in population $i \in I$ is given by
\begin{equation*}
    \hat{q}_x^i(p,y) = \dfrac{\exp\big(W^i_x(y - p_x)\big)}{\sum_{z\in Z} \exp\big(W^i_z(y - p_z)\big)}.
\end{equation*}
for some increasing functions $W^i_x$. 
\begin{proposition}\label{prop:logit_intro}
    Assume Logit demand, and suppose that under $p'$ the price for good $x$ increases, holding other prices fixed. If $\exp(W_z^i(\cdot))$ is convex for all $i \in I$ and $z \neq x$, and at the original prices the share of consumers purchasing good $x$ is sufficiently high, then $\CV_{med}(p') \leq \CV_{avg}(p')$.
\end{proposition}

\Cref{prop:logit_intro} (which is a special case of \Cref{prop:logit} proved below) shows that for Logit demand the median compensating can be systematically below the mean. In particular, this result applies to the widely used Cobb-Douglas formulation, as in \cite*{berry1995automobile}. In \Cref{sec:empirical}, we provide an empirical illustration applied to the analysis of non-tariff trade policy in \cite*{berry1999voluntary}.

These considerations suggest that the median compensating variation may be an appropriate welfare measure in many applied settings. However, unlike $\CV_{agg}$ and $V$, we show that the social preference induced by the median compensating variation violates Harsanyi's axiom, and thus does not admit a utilitarian representation (\Cref{lem:median_negative}). This naturally raises the question: what is the behavioral content of the social preferences induced by $\CV_{med}$? Can this measure be justified via some normatively appealing social preference? We now turn to providing such a foundation.

\section{General model}\label{sec:model}

\paragraph{Utilities over policies.} We have until now focused on welfare comparisons between two policies; but it is clear how to extend these comparisons to a more general problem of choosing among a menu of different policies. For our purposes, what is important is that there is a primitive notion of individual and collective welfare associated to each policy. In particular, for our main proposal in Theorem~\ref{thm:main}, we shall need individual agents' choice behavior.

In \Cref{sec:distributional} we return to the discussion in \Cref{sec:motivation}. We argue that our findings are relevant for the welfare comparisons among a given policy and its alternatives, as is the focus of applied welfare economics.

\paragraph{Notation and preliminary definitions.}
We write $\Delta(S,\Sigma)$ for the set of (countably additive) probability measures on a measure space $(S,\Sigma)$. If $S$ is a topological space, we denote by $\B_S$ the collection of Borel sets on $S$. When $A$ is a subset of a linear space, we denote by $\cvh(A)$ the \df{convex hull} of $A$, defined as the intersection of all convex sets that contain $A$. If $x,y\in\Re^n$, we denote the inner product $\sum_{i=1}^n x_i y_i$ by  $x\cdot y$. 
\paragraph{Preferences over policies.} Let $\prizes$ be a finite set of $\dimU+1$ \df{policies}, and let $\Delta(\prizes,2^\prizes)$ be the set of all \df{lotteries} on $\prizes$. In \Cref{sec:motivation}, we identified policies in $\prizes$ with their consequences for prices. Choosing among policies is then equivalent to choosing among price vectors. Here we consider policies in the abstract, and we allow for  random choices over policies. An agent who is choosing among lotteries is then choosing an element of  $X\coloneqq \Delta(\prizes,2^\prizes)$, which we identify with the $\dimU$-dimensional \df{simplex} $\{x\in\Re^{\dimU+1}:x\geq 0 \text{ and } \sum_{i=1}^{\dimU+1}x_i=1\}$. Expected utility preferences are given by a function $u:\prizes\to\Re$, the von Neumann-Morgenstern (vNM) utility, which we may view as a vector $u\in \Re^{\dimU +1}$, so that lottery $x$ is preferred to $y$ if and only if $u\cdot x\geq u\cdot y$. Since a vNM utility is only unique up to a positive affine transformation, we may without loss of generality focus on utilities, $u$, which are normalized so that the utility of the $\dimU + 1$-st prize, $u_{\dimU+1}$, is zero. Let then $U\coloneqq \{u\in\Re^{\dimU+1}:u_{\dimU+1}=0\}$. Given $u\in U$, a lottery $x\in X$ has expected utility $u\cdot x\coloneqq \sum_{i=1}^{\dimU}u_ix_i$. 

We follow Harsanyi in assuming that choices over lotteries may be random and working with expected utility preferences. In Section~\ref{sec:further_discussion}, we argue that there are good reasons to allow for lotteries over policies. For now, let us reiterate Harsanyi's argument. He justifies the assumption through a hypothetical risky choice ``behind a veil of ignorance.'' A choice between societal outcomes, policies, made in an original position, before each agent knows the place that they occupy in society. Such a choice is made in the  presence of risk, and lends itself to expected utility analysis. A crucial advantage is that the vNM utilities are cardinally meaningful. In the words of \cite{harsanyi53}: ``the analysis of impersonal value judgments concerning social welfare seems to suggest a close affinity between the cardinal utility concept of welfare economics and the cardinal utility concept of the theory of choices involving risk.''\footnote{A similar point is made by \cite{zeckhauser1969majority} in arguing for lotteries over social outcomes.}

\paragraph{Random expected utility.} A \df{decision problem} is a finite set $D\subseteq X$ of lotteries on $\prizes$. Let $\D$ denote the set of all decision problems. A \df{stochastic choice rule} is a function $\rho:\D\to \Delta(X,\B)$ with  $\rho(\cvh(D),D)=1$. We can interpret $\rho$ as representing either the (potentially stochastic) choices of an individual, or the proportion of individuals in a population who choose each of the alternatives in $D$, where each individual's choice is deterministic. Unless otherwise noted, we maintain the stochastic-choice interpretation for the purpose of discussion. 

Our model of discrete choice is the random expected utility (REU) theory of \cite{gul2006random}. It can be thought of as a random coefficients discrete choice model. To formulate a random expected utility, we need the pertinent measurability structure. Let $N(D,x)\coloneqq\{u\in U:u\cdot x\geq u\cdot y \ \ \forall \ y\in \cvh(D)\}$ denote the set of all utility indexes that are maximized at $x\in D$, and $N^+(D,x)\coloneqq\{u\in U:u\cdot x > u\cdot y \forall y\in  \cvh(D)\setminus \{x\}\}$ those that achieve a unique maximum at $x$ in $D$. Now let $\F$ be the algebra, and $\F^*$ the $\sa$-algebra, generated by the sets $N(D,x)$.

A \df{random expected utility} is a countably additive probability measure $\pi\in \Delta(U,\F^*)$. A random expected utility $\pi$ is \df{regular} if $\pi(N(D,x))=\pi(N^+(D,x))$ for all $(D,x)$, with $x\in D$. That is, indifferences occur with zero probability. Such $\pi$ generates a stochastic choice rule
\[
\rho_{\pi}(x,D) \coloneqq \pi(N(D,x))
= \pi(\{u : u\cdot x\geq u\cdot y  \text{ for all }y\in \cvh(D) \}),
\] for $x\in D$. 

We assume below that agents have regular random expected utilities.\footnote{Under the population interpretation of random utility, this means that there are no atoms in the distribution of preferences in the population.} However, we do not want to impose this condition exogenously on the planner. For one thing, we want to leave open the possibility that the planner has a deterministic utility, which is a special case of REU but violates regularity. Without regularity, indifferences can occur with positive probability, and so some tie-breaking rule is needed to define the choice rule. For our purposes, it is not necessary to make the tie-breaking rule explicit. Therefore, we stretch the use of notation and represent the choice rule for non-regular $\pi$ using $\rho_{\pi}$ as well.

\section{Collective choice}\label{sec:collective}

\subsection{Harsanyi's justification of utilitarianism}\label{sec:harsanyijustification}

We first present a suitably adapted version of Harsanyi's approach to utilitarianism. It applies to the population interpretation of the discrete choice models.  Consider a unit-mass non-atomic population of individuals indexed by $i \in [0,1]$, and a function $i\mapsto u_i \in \Re^n$ that assigns each individual agent a vNM utility over lotteries. A planner also has a vNM utility $u$. The standard practice in applied welfare analysis presumes knowledge of the population distribution over vNM utilities $u_i$, and, for purposes of policy evaluation, takes the planner's $u$ to be the expected value of these individual utilities.

To be more specific, we need to make a few additional assumptions. Suppose that the mapping $i\mapsto u_i$ is continuous. Let $\mathcal{B}$ denote the Borel $\sa$-algebra on $[0,1]$, and suppose that $\mu\in\Delta([0,1],\mathcal{B})$ describes the distribution of preferences in the population. The utilitarian approach, represented by both the Bergson-Samuelson and average-compensating-variation measures discussed in \Cref{sec:motivation}, is to take
\begin{equation}\label{eq:expU}
u = \int_I u_i \diff \mu(i)
\end{equation}

We first review the rationale behind using~\eqref{eq:expU}. This specification is consistent with Harsanyi's (\citeyear{harsanyi1955cardinal}) theorem,  but we will see that there are some fundamental problems in providing a justification for~\eqref{eq:expU}, even in a setting where Harsanyi's analysis is viable.

\citeauthor{harsanyi1955cardinal} assumes a finite population of agents. We apply a  generalization of his theorem to infinite populations due to \cite{zhou1997harsanyi}. The generalization requires one additional, arguably mild, assumption: Suppose that each $u_i$ can be normalized so that there are two lotteries $x_0$ and $x_1$ with $u_i\cdot x_0=0<u_i\cdot x_1$ for all $i$. This assumption entails the existence of a common direction of utility improvement (we may think of this as an objective notion of monotonicity).\footnote{A version of Harsanyi's theorem is possible without this assumption, and even without any topological assumptions on the space of agents, but it is even further from~\eqref{eq:expU} than Theorem~\ref{thm:zhou} below.}

We say that the planner's utility $u$ has the \df{Pareto property} if, whenever $x\cdot u_i\geq y\cdot u_i$ for all $i\in I$, then  $x\cdot u\geq y\cdot u$. In words, when all individual agents rank $x$ (weakly) over $y$, then so does the planner. Zhou's version of Harsanyi's theorem is:

\begin{theorem}[\citealt{zhou1997harsanyi}]\label{thm:zhou}The planner's utility $u$ has the Pareto property if and only if there exists $\la\in\Delta([0,1],\mathcal{B})$ and a scalar $\kappa>0$ such that \[
u = \kappa \int_I u_i \diff \la(i).
  \]
  \end{theorem}

Now, there are a few important differences between the use of \eqref{eq:expU} and the representation provided by Theorem~\ref{thm:zhou}. The scalar $\kappa$ is a difference, but it does not matter since we may as well take the planner's utility to be $\frac{1}{\kappa}u$ without changing the planner's preferences. The important difference lies in the use of $\la$ instead of $\mu$: these measures may not be related in any way.

To fix ideas, consider a simple special case in which $I$ is finite, with $\abs{I}=\nagents$ and $\mu$ being uniform. Then $\eqref{eq:expU}$ becomes $u=\frac{1}{\nagents}\sum_{i\in I}u_i$, while Theorem~\ref{thm:zhou} says that there exist non-negative weights $\la(i)$ with $u=\sum_{i\in I}\la(i)u_i$. Theorem~\ref{thm:zhou} cannot guarantee an equal weight on each utility function, \emph{nor could it without attaching cardinal meaning to the scale of utilities $u_i$}.\footnote{An equal-weight representation is possible under different primitives, but with an endogenously chosen vNM utility for each agent. See the foundation for utilitarianism in \cite{harsanyi53,harsanyi78,Harsanyi1979} (and the clarifying discussion by \cite{Weymark1991}).} In the discrete-choice framework, we only have access to choice data from each individual agent. So even if we imagine that there is a role for the scale of utility, perhaps a notion of preference intensity, the choice theoretic starting point of the theory presumes no access to information on preference intensity.

\subsection{From utilities to choices}

We now turn to the main methodological proposal in our paper, which is achieved by aggregating choices instead of utilities. We have emphasized how, in aggregating utilities, we are forced to make interpersonal cardinal utility comparisons. Here we depart from \cite{harsanyi1955cardinal} by studying utilitarian \textit{social choice} rather than social preferences. Aggregating social choices instead of preferences has a long history in economics: for example,  \cite{senQuasiTran1969,sen_collectivechoice1970} recreates a version of Arrow's theorem when the planner's choice behavior is modeled instead of an aggregate preference. In our model, thanks to the linear structure afforded by expected utility theory, it makes sense to linearly aggregate agents' choices; a natural choice-based analogue to utilitarianism. We may linearly aggregate choices without having to measure utilities on a common scale.

Let $\pi_i$, $1\leq i\leq m$, describe a society of $m$ agents, each endowed with a regular random expected utility. Consider a planner with random expected utility $\pi$. Importantly, $\pi$ may not be regular; it could, in principle, assign probability one to a single vNM utility. Let $\rho_{\pi_i}$ and $\rho_{\pi}$ be the agents' and the planner's stochastic choice functions. Our definition of utilitarianism entails the linear aggregation of agents' stochastic choice functions: The planner $\pi$ is a \df{weighted utilitarian for} $\pi_1,\ldots,\pi_m$ if there exist weights $\al_1,\ldots,\al_m\geq 0$, $\sum_{i=1}^m\al_i=1$, such that $\pi = \sum_{i=1}^m\al_i \pi_i$.

We introduce a Pareto axiom that, like our notion of utilitarianism, is based on choices. Suppose that every agent, facing a menu $D$ that includes lotteries $x$ and $y$, chooses $x$ more often than $y$. Then our Pareto criterion demands that the planner must also choose $x$ more often than $y$. Going further, suppose that every agent is more than twice as likely to pick $x$ as $y$ from $D$. Our Pareto axiom then requires the planner to reflect that strength of agreement: the planner must be at least twice as likely to choose $x$ as $y$. Agents can also agree on subtler, comparative statements. For instance, they might all agree that the gap between the probabilities of choosing $x$ and $y$ is smaller than the gap between the probabilities of choosing $z$ and $y$. Our axiom says: whenever all agents agree on any such linear constraint on choice probabilities from $D$, the planner’s choices must obey the same constraint.

Formally, $\pi$ \df{respects Pareto} in relation to $\pi_1,\ldots,\pi_m$ if, for any decision problem $D$ and any function $c:D\to\Re$, if $\E_{\rho_i}c\geq 0$ for every agent $i$, then the planner’s choice probabilities $\rho$ must also satisfy $\E_{\rho}c\geq 0$.

A local relaxation of utilitarianism turns out to be interesting. The planner $\pi$ is a \df{local behavioral utilitarian for} $\pi_1,\ldots,\pi_m$ if, for every decision problem $D$ there are weights $\al^D_1,\ldots,\al^D_m\geq 0$, $\sum_{i=1}^m\al^D_i=1$, so that $\rho_\pi(x,D) = \sum_{i=1}^m\al^D_i \rho_{\pi_i}(x,D)$ for all $x\in D$. It should be easy to see, by a separating-hyperplane argument, that local behavioral utilitarianism and respecting Pareto are equivalent. It would also seem that the local property is substantially weaker than what we have called (global) utilitarianism. Our main theory will actually show that they are equivalent.

A fourth property of interest is based on an arbitrary choice rule $\rho$. We can say that the \df{agreement} between $\rho$ and a random expected utility $\hat \pi$ is the probability, given by
\[
\sum_{x\in D}\rho(x,D)\rho_{\hat \pi}(x,D),
\] that they choose the same element of $D$ if they are choosing independently. 

Our main result is a characterization of utilitarian social choice. 

\begin{theorem}\label{thm:main} The following statements are equivalent:
  \begin{enumerate}
  \item $\pi$ is a weighted utilitarian for $\pi_1,\ldots,\pi_m$.
  \item $\pi$ is a local behavioral utilitarian for $\pi_1,\ldots,\pi_m$.
  \item $\pi$ respects Pareto in relation to $\pi_1,\ldots,\pi_m$.
  \item For any stochastic choice $\rho$, and any $D\in\D$, the agreement between $\pi$ and $\rho$ is not more than the largest and not less than the smallest agreement between $\rho$ and $\pi_i$, $1\leq i\leq m$.
    \end{enumerate}
  \end{theorem}

The main insight of \Cref{thm:main} is that the Pareto property for social choice implies that the planner is a weighted utilitarian. In particular, since each $\pi_i$ is regular, this implies that $\pi$ is as well, so the planner's choice rule must be stochastic. The role of randomness in social choice is discussed further in \Cref{sec:further_discussion}.  

The proof of Theorem~\ref{thm:main} is in \Cref{proof:main}. Here is a high-level summary of the main ideas in the proof. The property of being a weighted utilitarian is equivalent to $\pi$ belonging to the convex hull of $\pi_1,\ldots,\pi_\nagents$; so by a separation argument, we can say that $\pi$ is a weighted utilitarian if and only if there is no continuous linear functional that separates the planner's $\pi$ from each of the individual $\pi_i$. By the measurability assumptions of \cite{gul2006random}, we may approximate such a linear functional by means of a linear combination of step functions, each of them being the indicator function of a convex cone. Now, connecting these cones with the set of vNM indexes that choose some $x$ out of a decision problem $D$, we observe that the evaluation of a step function at each $\pi$ and $\pi_i$ coincides with the probability of choosing $x$ from $D$ (this step needs, for technical reasons, a second approximation argument). Finally, the separation property translates into being a local behavioral utilitarian. 

Perhaps the most basic insight in the proof is that the global property of being a weighted utilitarian is equivalent to a local property holding at a particular decision problem $D$: the reason for this equivalence is that the linear functional in the separating argument can be approximated through a simple function built from a finite set of disjoint cones, and these may be ``glued'' into a single decision problem $D$ (more descriptively, attached to $D$ as the normal cones at different extreme points of $D$).

The proof in \Cref{proof:main} deals with two mathematical difficulties. The linear spaces we work with have no tractable topological duals. We work with the space of charges, and use an approach proposed by \cite{rao1983theory} that allows us to approximate the separating linear functional. As far as we know, we are the first to apply this technique in economics. Second, we believe that any collection of disjoint pointed cones can be obtained from a single decision problem, but we have not been able to prove this. So there is a second approximation used in getting the relevant decision problem. This second approximation requires the monotone continuity property of countably additive probabilities, hence our use of countably additive random expected utility.

Conditions (3) and (4) in \Cref{thm:main} are both intuitively related to the alignment between the planner's choice rule and those of the agents. We formalize this intuition in \Cref{sec:alignment}, which implies an alternative characterization of weighted utilitarianism. 

 \subsection{Harsanyi's utilitarianism and the Pareto axiom for social choice}\label{sec:theparetoaxiom}

 The approaches to welfare comparisons based on Harsanyi, as outlined in \Cref{sec:motivation}, violate the ``respecting Pareto'' principle in Theorem~\ref{thm:main}. Suppose that  $x \in X$ is evaluated according to \df{expected utilitarian welfare} (EUW) 
 \begin{equation*}
     V(x) = \sum_{i =1}^m \alpha_i \E_{\pi_i} u \cdot x 
 \end{equation*}
 for some weights $\alpha_1, \dots, \alpha_m \geq 0$, $\sum_{i=1}^m \alpha_i = 1$. This is essentially the Bergson-Samuelson approach. In other words, social preferences have an expected utility representation, $V(x) = v \cdot x$ where $v := \sum_{i =1}^m  \alpha_i\E_{\pi_i}u$. This is clearly a special (degenerate) case of REU.\footnote{It is not regular, but this distinction is not important here.}

 Now we argue that this criterion violates our axiom. Consider the following simple example. Two agents, $\{1,2\}$, face a binary choice between alternatives $D = \{a,b\}$. Adopt the population interpretation, so that an ``agent'' here refers to a unit mass population of individuals. Normalize the utility of alternative $b$ to zero. In population 1, a fraction $0.9$ of the agents have utility $1$ for alternative $a$, and the remainder have utility $-1$. In population $2$ a fraction $0.7$ have utility $-1$ for alternative $a$, and the remainder have utility $1$. That is, population $1$ tends to favor alternative $a$ and population $2$ alternative $b$, but the tendency is more pronounced in population 1. 

 An EUW social choice rule which places equal weight on all individuals, $\rho^*$, will select alternative $a$ with probability 1. In contrast, an equally-weighted utilitarian rule, $\rho$, will select alternative $a$ with probability $.6$. The EUW social choice rule evidently does not respect Pareto. For example, consider $c(a) = 1$ and $c(b) = 2$. Then
 \begin{equation*}
     E_{\rho^*(\cdot,D)} c = 1 \leq E_{\rho_{\pi_1}(\cdot,D)} c = 1.1 \leq E_{\rho(\cdot,D)} c  = 1.4 \leq E_{\rho_{\pi_2}(\cdot,D)} c = 1.7
 \end{equation*}
 Why should we care? The EUW rule inherits the intuitive appeal of \cite{harsanyi1955cardinal}, since under the population interpretation there is in fact no stochasticity in individual choice. Intuitively, it chooses the alternative that ``maximizes the average utility of the population'': the average utility of alternative $a$ is $.9 + .3 - (.1 + .7) = .4$ (while that of $b$ is normalized to zero). 

 However, this interpretation of EUW implicitly takes a stand on interpersonal utility comparisons. Suppose, for example, that preferences in population 1 are as before, but those in population 2 are in fact more intense: in this group a fraction $0.7$ have utility $-3$ for alternative $a$, and the remainder have utility $3$. Then the average utility of $a$ is $0.9 + 0.3\cdot 3 - (0.1 + 0.7 \cdot 3) = -0.4$, so alternative $b$ maximizes the average utility in the population. Since all we have done is apply an affine transformation to the utility function, such preference intensity can never be identified from choice data; instead, the choice of utility scale is entirely at the discretion of the social planner or analyst. 

If the planner believes that preference intensity has real meaning, then they face a dilemma. On the one hand, they may wish to use a social choice rule that is responsive to preference intensity. However, this is impossible due to the inherent lack of identification. On the other hand, the planner can simply impose a utility scale for the purposes of aggregating preferences, for example, by normalizing utilities to lie between $-1$ and $1$ and aggregating with equal weights as above. In this case, however, we cannot claim that the EUW rule ``maximizes average utility'' in any meaningful sense. 

Faced with the identification problem, there is only one real way out: adopt a social choice rule that stays sound even when we are unsure about preference intensities. Weighted utilitarianism delivers exactly this kind of robustness.\footnote{In fact, weighted utilitarianism is the only social choice rule that is robust to uncertainty about preference intensity and satisfies an additional condition that says, roughly, that the social choice rule must outperform a dictatorship. We omit the details, as this result is not central to our analysis. The main point is that regardless of the true preference intensities, a weighted-utilitarian rule that places strictly positive weights on every group is guaranteed to achieve a fraction of the optimal utilitarian welfare that is bounded away from zero. In contrast, the EUM rule can do arbitrarily poorly in this regard.} The planner's choice of weights might still be informed by some judgment of individuals' preference intensities. But, unlike the EUM rule, the \emph{interpretation} of the weighted-utilitarian rule depends only on the weights placed on each group and not on assumptions regarding relative preference intensities. As we show in the subsequent section, this feature of weighted-utilitarian social choice translates into welfare assessments based on quantiles of individual-level effects rather than averages. 

\section{Distributional welfare analysis}\label{sec:distributional}

We now discuss the implications of our characterization for applied welfare analysis. In particular, we return to the connection between our characterization and the discussion of welfare in \Cref{sec:motivation}. To formalize this connection, we first need to adapt the welfare measures discussed in \Cref{sec:motivation} to lotteries over price vectors: in the setting of our general model, price vectors play the role of policies, and the objects of choice are lotteries over price vectors. 

Preferences over lotteries determine, in particular, a preference over deterministic policies. Conversely, a preference over deterministic policies will have one or more \df{extensions} to a preference over lotteries. The social preferences induced by $\CV_{agg}$ and $\CV_{avg}$, which were defined over deterministic policy changes, admit a natural extension to the set, $\Delta(P)$, of lotteries over price vectors. In the case of aggregate compensating variation, for $x',x'' \in \Delta(P)$ we say that $x'' \succsim_{agg} x'$ if and only if
\begin{equation*}
    \E_{p \sim x''} \left[ \int v(p,y|\varepsilon) \hat{F}(d\varepsilon) \right] \geq \E_{p \sim x'} \left[ \int v(p,y|\varepsilon) \hat{F}(d\varepsilon) \right],
\end{equation*}
and for average compensating variation, $x'' \succsim_{avg} x'$ if and only if
\begin{equation*}
    \E_{p \sim x''} \left[ \int \CV(p,p^0|\varepsilon) \hat{F}(d\varepsilon) \right] \geq \E_{p \sim x'} \left[ \int \CV(p,p^0|\varepsilon) \hat{F}(d\varepsilon) \right].
\end{equation*}
In both cases, the representation takes the standard utilitarian form, as in Equation~\eqref{eq:expU}, and can be justified by appealing to Harsanyi's axioms, as in Theorem~\ref{thm:zhou}. Of course, the general problems with such a justification, related to obtaining the needed weights as discussed above, continue to be valid. Note that  the restrictions of $\succsim_{agg}$ and $\succsim_{avg}$ to the set of degenerate lotteries coincide with the original measures from \Cref{sec:introduction}, as desired. 

It is less clear how to extend the preferences defined by $\CV_{med}$ to non-degenerate lotteries over price vectors. At a minimum, however, it seems natural that such an extension should satisfy the following consistency condition. 

\begin{definition}
    The extension $\succsim_{med}$ is \textit{weakly consistent} if for any binary lottery $x \in \Delta(P)$ supported on $\{p^1,p^2\}$, if $\CV_{med}(p^1) < 0$ and $\CV_{med}(p^2) < 0$ then $p^0 \succ_{med} x$.\footnote{Where $p^0 \succ_{med} x$ if $p^0 \succsim_{med} x$ and $\neg(x \succsim_{med} p_0)$. We follow the usual convention of identifying a price vector $p$ with the lottery that is degenerate on $p$.} 
\end{definition}

In other words, if median compensating variation is negative for both possible realizations of the alternative price vector, then the extension $\succsim_{med}$ should rank $p^0$ strictly above this binary lottery. Surprisingly, weak consistency is incompatible with the Pareto property. 

\begin{proposition}\label{lem:median_negative}
    If $\succsim_{med}$ satisfies weak consistency, then it violates the Pareto property.
\end{proposition}

The proof of \Cref{lem:median_negative} can be found in \Cref{proof:median_negative}.

Given the message of \Cref{lem:median_negative}, the justification for using $\CV_{med}$ is unclear. 
It turns out, however, that weighted-utilitarian social \textit{choice} (rather than preference)  provides a foundation for this measure; we just need to define the appropriate notion of compensating variation for stochastic choice rules. Let $P$ be the set of price vectors, $Y = \R$ the set of income levels, and $\Delta(P\times Y)$ the set of lotteries over prices and incomes.\footnote{Formally, our analysis of stochastic choice applies to lotteries over a finite set of prizes, whereas here the set of prizes is infinite. To address this discrepancy we could discretize the set of price vectors and income levels. However for the sake of clarity we work directly with $\Delta(P\times Y)$ here.} Let $\mathcal{D}^{PY}$ be the corresponding set of decision problems. In a stochastic choice setting, the natural notion of indifference between two lotteries is that the decision maker is equally likely to choose either one when faced with a binary decision problem. Using this notion we can extend the definition of compensating variation to stochastic choice rules. 

\begin{definition}
    Given a stochastic choice rule, $\rho$, on $\mathcal{D}^{PY}$, define the \textit{stochastic compensating variation}, $\CV^{\rho,\frac{1}{2}}$, for a shift from $p^0$ to $p'$ by
    \begin{equation*}
        \CV^{\rho,\frac{1}{2}}(p^0,p') = \inf\left\{x : \rho\big(\{(p^0,y), (p',y-x) \}\big) \geq \frac{1}{2} \right\}.
    \end{equation*}
\end{definition}

In words, under monotonicity of choice probabilities in money, $\CV^{\rho,\frac{1}{2}}(p^0,p')$ is the income reduction $x$ such that when faced with the decision problem consisting of the two degenerate lotteries $(p^0,y)$ and $(p',y - x)$, the alternatives are chosen with equal probability. If we adopt the population interpretation of stochastic choice, we may say that after the reduction $\CV^{\rho,\frac{1}{2}}(p^0,p')$, half of the population of agents prefer $p'$ over $p^0$, while the other half has the opposite preference. Thus $\CV^{\rho,\frac{1}{2}}(p^0,p')$ can be thought of as a democratic notion of compensating variation.  

Return now to the question of evaluating the policy shift from $p^0$ to $p'$. Recall that the population consists of a unit mass of consumers, categorized into a finite set of types, $I$, based on observable characteristics. Let $\sigma = (\sigma_i)_{i\in I}$, where $\sigma_i$ is the fraction of the population with type $i$. The distribution over preferences for this population is estimated to be $\hat{F}_i$. Denote by $\hat{\pi}_{\alpha}$ the distribution over preferences for a weighted utilitarian social planner with weights $\alpha_i$, i.e. $\hat{\pi}_{\alpha}(\varepsilon) = \sum_{i\in I} \alpha_i \hat{F}_i(\varepsilon)$. Let $\rho^{\alpha}$ be the corresponding stochastic choice rule. Observe that when $\sigma_i=\alpha_i$ for all $i\in I$, then the social planner weights types according to their representation in society, meaning that all individuals are weighted equally. It turns out that the median compensating variation in the population is also the stochastic compensating variation of a weighted utilitarian social planner who weights all consumers equally. 

\begin{proposition}\label{prop:foundation}
$\CV^{\rho^{\sigma},\frac{1}{2}}(p^0,p') = \CV_{med}(p^0,p')$. 
\end{proposition}
The proof of \Cref{prop:foundation} is in \Cref{proof:foundation}.

In other words, the social preferences over alternative price vectors, relative to the baseline of $p^0$, induced by $\CV_{med}$ are exactly those of a weighted utilitarian social planner who weights consumer types according to their share of the population.

It is straightforward to extend this argument to provide a foundation for a wider class of distributional welfare measures. Given a distribution $\pi$ over preferences, let $G_{\pi}(\cdot|p^0,p')$ be the CDF of the distribution of compensating variation, and let $T_{\pi}(\tau|p^0,p')$ be the $\tau$-percentile of this distribution.\footnote{That is, $T_{\pi}(\tau|p^0,p') = \inf\{ z \in \R : G_{\pi}(z|p^0,p') \geq \tau\}$.}

\begin{definition}
    Given a pair $(\tau, \alpha)$; where $\tau \in [0,1]$ and $\alpha \in \Delta(I)$; define the \textit{distributional compensating variation} $\CV_{\alpha,\tau}(p^0,p') := T_{\hat{\pi}_{\alpha}}(\tau|p^0,p^1)$. 
\end{definition}

That is, $\CV_{\alpha,\tau}(p^0,p')$ is the $\tau$-percentile of the compensating variation distribution after we distort the distribution of preferences by $\alpha$ (so with this notation $\CV_{med}(p^0,p') =  \CV_{\sigma,\frac{1}{2}}(p^0,p')$). This notion corresponds to a modification of stochastic compensating variation for the planner's preferences.

\begin{definition}
    Given a stochastic choice rule, $\rho$, on $\mathcal{D}^{PY}$ define
    \begin{equation*}
        \CV^{\rho,\tau}(p^0,p') = \inf\left\{x : \rho\big((p^0,y),\{(p^0,y), (p',y-x) \}\big) \geq \tau \right\}.
    \end{equation*}
\end{definition}

In words, $\CV^{\rho,\tau}(p^0,p')$ is the smallest payment under prices $p'$ that can be extracted such that the probability of choosing $(p^0,y)$ over $(p', y-x)$ is at least $\tau$. As $\CV^{\rho,\tau}(p^0,p')$ is increasing in $\tau$, setting $\tau < \frac{1}{2}$ is analogous to a supermajority requirement for moving away from the status quo. By essentially the same proof as \Cref{prop:foundation}, we have the following. 

\begin{proposition}\label{prop:foundation_general}
    $\CV^{\rho^{\alpha},\tau}(p^0,p') = \CV_{\alpha,\tau}(p^0,p')$.
\end{proposition}
The measure $\CV_{\alpha, \tau}$ thus provides a principled way to incorporate distributional considerations into welfare analysis. The interpretation of the weights $\alpha$ is discussed below. 

We emphasize that distributional welfare measures are indeed agnostic about interpersonal utility comparisons. For clarity, we focus this discussion on $CV_{med}$. Under this measure, prices $p'$ are deemed preferable to $p^0$ if at least half the population prefers $p'$ to $p^0$. The median does not take a stand on preference intensities the way that the mean does. It is important here to note the distinction between taking the median of welfare changes as opposed to levels. Comparing two price regimes based on median welfare levels would imply, in particular, that permuting the welfare levels of individuals under prices $p'$ does not affect the social ranking of $p'$ and $p_0$. This is a form of interpersonal comparison of preference intensity. However $CV_{med}$ is the median of a welfare \textit{change}. This measure implies that making person $x$ better off is just as good, from a social perspective, as making person $y$ better off. This is an interpersonal comparison of ordinal, but not cardinal, preferences.

Before proceeding with the discussion of the welfare measures we first briefly comment on the issue of Scitovsky reversals \citep{graaff1967theoretical}. While these are not central to the current discussion, the point is worth clarifying in order to avoid confusion.

\begin{remark}\label{remark:scitovsky}  We focus here on comparing alternative price vectors, $p'$, to a baseline $p^0$; each of the three measures discussed in \Cref{sec:motivation} induces a cardinal representation of preferences over alternative price vectors, but we are primarily interested in separating the price vectors that are better than $p^0$ from those that are worse than $p^0$. We could also use these measures to compare two price vectors $p',p'' \neq p^0$. However the ranking induced by $\CV_{avg}$ over such $p',p''$ is not necessarily consistent with comparing $p'$ and $p''$ directly using average compensating variation. That is, $CV_{avg}(p'') \geq \CV_{avg}(p')$ does not necessarily imply that 
    \begin{equation*}
       \int \CV(p',p''|\varepsilon)\hat{F}(d\varepsilon) \geq 0.
   \end{equation*}
   This discrepancy arises because of wealth effects, in the presence of which, as is well-known, $\CV(p',p''|\varepsilon)$ may not be equal to $\CV(p^0,p''|\varepsilon) - \CV(p^0,p'|\varepsilon)$. Similarly, $\CV_{med}(p'') \geq \CV_{med}(p')$ does not necessarily imply that $\text{Median}_{\varepsilon\sim \hat{F}}[ \CV(p',p''|\varepsilon)] \geq 0$. On the other had, it is easy to see that comparisons using $\CV_{agg}$ do not suffer from this discrepancy.
\end{remark}

\subsection{Interpreting the weights}

It is straightforward to define an $\alpha$-weighted version of $\CV_{avg}$ by taking the expectation of compensating variation under the distribution $\hat{\pi}_{\alpha}$. If consumers have quasi-linear utility, so that payoffs are denoted in monetary terms, then the weights $\alpha$ can be interpreted as the planner's rate of substitution between monetary transfers to different types of agents (see \cite{backus2024surplus} for further discussion). That is, $\frac{\alpha_i}{\alpha_j}$ is the maximum amount of money the planner would be willing to extract from a type-$j$ consumer in order to transfer one unit to a type-$i$ consumer. However with wealth effects the interpretation is less straightforward, unless we restrict attention to local perturbations as in \cite{saez2016generalized}.

On the other hand, the weights for $\alpha$-weighted distributional compensating variation have a straightforward interpretation, whether or not there are wealth effects. We can think of the unweighted median compensating variation, $\CV_{med}$, as counting the number of consumers who are made better off by the price change, after a suitable monetary payment has been extracted, and $\CV_{med}$ is the smallest payment such that no more than half of the consumers are strictly better off. Similarly, the measure $\CV_{\alpha,\tau}$ counts the number of consumers who are made better off by the price change, but does not necessarily give them equal weight: making a type-$i$ consumer strictly better off is counted the same as making $\frac{\alpha_i}{\alpha_j}\frac{\sigma_j}{\sigma_i}$ type-$j$ consumers strictly better off. For example, consider the set of consumers whose compensating variation is strictly above $\CV_{med}(p^0,p')$, i.e. the set of consumers made strictly better off if a payment of $\CV_{med}(p^0,p')$ is extracted after the price change. Suppose that all these consumers are of type $j$. Then if we set $\alpha_i < \sigma_i$, the distributional comparative advantage measure $\CV_{\alpha, \frac{1}{2}}(p^0,p')$ will be less that $\CV_{med}(p^0,p^1)$. 

From an applied perspective, this interpretation makes the exercise of choosing the weights relatively straightforward. The analyst choosing the weights needs to answer the question: it is just as good to make one type-$i$ consumer better off as how many type-$j$ consumers? The answer to this question pins down $\frac{\alpha_i}{\alpha_j}\frac{\sigma_j}{\sigma_i}$. Note however that the welfare analysis in this case is not overly sensitive to the choice of weights; starting from $\alpha$, any re-shuffling of weights among the consumers whose compensating variation is above $\CV_{\alpha,\tau}$ does not change the value of the welfare measure, and similarly for those below. Thus compared to weighted average compensating variation, the welfare conclusions are relatively robust to the choice of weights.\footnote{The robustness to such perturbations is a defining feature of quantile-based aggregation. It is closely related to the Pivotal Monotonicity condition of \cite{rostek2010quantile}, which is a key axiom in behaviorally characterizing quantile maximization in the context of choice over uncertain acts.}

\subsection{Properties and estimation of distributional compensating variation}\label{sec:estimation}

Let $p^0$ be a status-quo price vector, and $p'$ an alternative price that would result from a policy change. Some goods may increase in price, while other decrease, in going from $p^0$ to $p'$.  Without loss of generality, we may label the goods so that 
\begin{equation}\label{eq:price_order}
    p'_Z - p^0_Z \geq p'_{Z-1} - p^0_{Z-1} \geq \dots \geq p'_1 - p^0_1.
\end{equation}
The demand function for type $i$ is defined as
\begin{equation*}
    \hat{q}^i_x(p,y) = \int 1\{ u(x, y - p_x|\varepsilon) > \max_{z \neq x} u(z,y - p_z| \varepsilon)\} F_i(d\varepsilon).
\end{equation*} In words, $\hat{q}^i_x(p,y)$ is the probability of choosing good $x$ when prices are $p$ and income $y$ or, using the preferred population interpretation in the IO literature, $\hat{q}^i_x(p,y)$ is the market share of good $x$. 

Given weights $\alpha$, define the \textit{$\alpha$-weighted aggregate demand} by $ Q^{\alpha}_x(p,y) := \sum_{i\in I} \alpha_i  \hat{q}^i_x(p,y)$. We are now in a position to calculate the distribution of the compensating variation in the population.

\begin{proposition}\label{prop:CV_cdf}
Let $\Delta p_x = p'_x - p^0_x$, where prices are labeled as in \cref{eq:price_order}. Then
\begin{equation*}
    G_{\hat{\pi}_\alpha}(a|p^0,p') =
    \begin{cases}
        1, \quad &\text{if } a \geq -\Delta p_1 \\
        1- \sum_{x = 1}^j &Q^{\alpha}_x(p'_1,\dots ,p'_j, p^0_{j+1} - a, p^0_{j+2} - a, \dots, p^0_{Z} - a, y-a)\\
          & \text{if } -\Delta p_j > a \geq -\Delta p_{j+1}, a < 0, \ 1 \leq j \leq Z-1\\
        1- \sum_{x = 1}^j &Q^{\alpha}_x(p'_1 + a,\dots ,p'_j + a, p^0_{j+1}, p^0_{j+2}, \dots, p^0_{Z}, y)\\
          & \text{if } -\Delta p_j > a \geq -\Delta p_{j+1}, a \geq 0, \ 1 \leq j \leq Z-1\\
        0, &\text{if } a < -\Delta p_Z.
    \end{cases}
\end{equation*}
\end{proposition}

The proof of \Cref{prop:CV_cdf} is in \Cref{proof:CV_cdf}. These expressions simplify when we restrict attention to price changes for a single good. 

\begin{corollary}\label{cor:single_price_cdf}
    If the price of good $Z$ increases and all other prices remain unchanged under $p'$, then
    \begin{equation*}
    G_{\hat{\pi}_\alpha}(a|p^0,p') =
    \begin{cases}
        1, \quad &\text{if } a \geq 0 \\
        Q^{\alpha}_Z(p^0_1,\dots ,p^0_{Z-1}, p^0_{Z} - a, y-a)
          & \text{if } 0 > a \geq -\Delta p_{Z}\\
        0, &\text{if } a < - \Delta p_Z.
    \end{cases}
\end{equation*}
 If the price of good $1$ decreases and all other prices remain unchanged under $p'$ then
 \begin{equation*}
    G_{\hat{\pi}_\alpha}(a|p^0,p') =
    \begin{cases}
        1, \quad &\text{if } a \geq - \Delta p_1 \\
        1 - Q^{\alpha}_1(p'_1 + a, p^0_{2}, \dots, p^0_{Z}, y)
          & \text{if } -\Delta p_1 > a \geq 0, \\
        0, &\text{if } a < 0.
    \end{cases}
\end{equation*}
\end{corollary}

We can use these expressions to compare the median compensating variation to the mean. The proof of \Cref{prop:single_price_median} is in \Cref{proof:single_price_median}.

\begin{proposition}\label{prop:single_price_median}
    If the price of good $Z$ increases and all other prices remain unchanged under $p'$, then 
    \begin{itemize}
        \item $\CV_{\alpha,\frac{1}{2}}(p^0,p') \geq \CV_{avg}(p^0,p')$ if
        \begin{enumerate}[i.]
            \item $a \mapsto Q_Z^{\alpha}(p^0_{-Z}, p_x^0 - a, y-a)$ is convex, and
            \item at prices $p'$ and income $y + \Delta p_Z$ the share of consumers purchasing good $Z$ is sufficiently low.
        \end{enumerate}
        \item $\CV_{\alpha,\frac{1}{2}}(p^0,p') \leq \CV_{avg}(p^0,p')$ if 
        \begin{enumerate}[i.]
            \item $a \mapsto Q_Z^{\alpha}(p^0_{-Z}, p_x^0 - a, y-a)$ is concave, and
            \item at prices $p^0$ the share of consumers purchasing good $Z$ is sufficiently high.
        \end{enumerate}
    \end{itemize}
    If the price of good $1$ decreases and all other prices remain unchanged under $p'$ then 
    \begin{itemize}
        \item $\CV_{\alpha,\frac{1}{2}}(p^0,p') \geq \CV_{avg}(p^0,p')$ if 
        \begin{enumerate}[i.]
            \item $a \mapsto Q_1^{\alpha}(p'_1 + a, p^0_{-1}, y)$ is concave, and
            \item at prices $p'$ the share of consumers purchasing good 1 is sufficiently high.
        \end{enumerate}
        \item $\CV_{\alpha,\frac{1}{2}}(p^0,p') \leq \CV_{avg}(p^0,p')$ if 
        \begin{enumerate}[i.]
            \item $a \mapsto Q_1^{\alpha}(p'_1 + a, p^0_{-1}, y)$ is convex, and
            \item at prices $p^0$ the share of consumers purchasing good 1 is sufficiently low.
        \end{enumerate}
    \end{itemize}
\end{proposition}

\Cref{prop:single_price_median} highlights two important features of distributional welfare analysis. First, the median may be systematically above or below the mean depending on the shape of demand. Second, the relevant population of consumers matters. If, for example, the good whose price has increased was purchased by very few consumers under the original prices, then there will be very little effect. The median compensating variation will be zero, regardless of the shape of demand.

An important special case is when demand takes the widely-used Logit form, given by
\begin{equation*}
    \hat{q}_x^i(p,y) = \dfrac{\exp\big(W^i_x(y - p_x)\big)}{\sum_{z=1}^X \exp\big(W^i_z(y - p_z)\big)}.
\end{equation*}
for some increasing functions $W^i_x$. It turns out that if the marginal value of money does not decrease too fast, i.e. income effects are not too strong, then the CV distribution for a price increase will be skewed to the right, such that the median is below the mean.

\begin{proposition}\label{prop:logit}
    Under logit demand, consider a price increase for good $Z$, holding all other prices fixed. If $\exp(W_x^i(\cdot))$ is convex for all $i$ and $x \neq Z$, and at prices $p^0$ the share of consumers purchasing good $Z$ is sufficiently high, then $\CV_{\alpha,\frac{1}{2}}(p^0,p') \leq \CV_{avg}(p^0,p')$ for any $\alpha \in \Delta(I)$. 
\end{proposition}
The proof of \Cref{prop:logit} is in \Cref{proof:logit}.

\Cref{prop:logit} makes the case for considering more than one measure of welfare. The average compensating variation masks disparate policy impacts, and may lead to overly optimistic evaluations of the aggregate welfare loss. The conditions of \Cref{prop:logit} are met, for example, for the Cobb-Douglas utility specification (e.g. \citealt{berry1995automobile}) where $W_x^i(y - p_x) = a \log(y - p_x) + L(i,x)$ for some parameter $a$ and function $L(i,x)$ of individual and product characteristics. To avoid numerical issues, this specification is often replaced in practice with the local approximation $\tilde{W}_x^i(y - p_x) = -\frac{a}{y} p_x + L(i,x)$ which, while not precisely of the form in \Cref{prop:logit}, yields similar estimates for demand, and thus for the distribution of compensating variation \citep{berry1999voluntary}.

\section{\texorpdfstring{Empirical application: the voluntary export restraint \\ policy of the 1980s}{Empirical application: the voluntary export restraint policy of the 1980s}}\label{sec:empirical}

We develop an empirical application to illustrate the practical implications of the main points in our paper (further details are presented in \Cref{sec:empiricalappendix}). We focus on one of the first uses of the ``BLP methodology'' \citep{berry1995automobile} to conduct policy analysis: the evaluation of the ``voluntary export restraint'' policy of the 1980s in \cite{berry1999voluntary}. In our discussion, we refer to the first paper as BLP, and to the second as ``the VER paper.''

 The voluntary export restraint policy sought to limit Japanese car imports to the United States by asking the Japanese government to impose a quota on its exports to the US. The policy was implemented by the Reagan government as protectionist support for the US automobile industry, which had been adversely affected by the recession in the late 70's and struggled to compete with efficient, cheap, Japanese imports. The policy was viewed as a less-aggressive alternative to imposing a tariff on Japanese cars and remained in place, essentially, from 1981 to 1992. The voluntary export restraints were implemented by the Japanese Ministry of Trade and Industry, which imposed quotas on the exports from Japanese manufacturers to the US. 

The VER paper's stated ``primary goal'' was ``to provide some econometric evidence on the welfare implications of VER.'' To this end, the VER paper estimated a structural BLP-style model and used counterfactual simulations to compute the mean compensating variation. In our version of the exercise, we focus on alternative measures of welfare. 

Following the VER paper, the utility to consumer $i$ from buying car $j$ is assumed to be
 \[ u_{ij}=\underbrace{x_j'\bar\beta+\xi_j}_{\delta_j}\underbrace{-\alpha\frac{p_j}{y_i}+\sum_k \sigma_k x_{jk}v_{ik}}_{\mu_{ij}}+\epsilon_{ij}:=u(j,y_i,p_j|\epsilon_i),
 \]
where $p_j$ and $x_j$ are, respectively, the price and a vector of attributes of car $j$.  The random error $\epsilon_{ij}$ is assumed to be drawn from a type-I error distribution. The parameters $v_{ik}$ are random coefficients (see BLP) and reflect individual-specific preferences for the different attributes a car may have.  The difference between this specification and the one in BLP is in the term $\alpha\frac{p_j}{y_i}$. The VER paper justifies the assumption as coming from the log-normality of the income distribution, but they also remark that it corresponds to a linear approximation to the earlier specification in BLP.

The VER paper argues that voluntary export restraints should be modeled, at the firm level, as an implicit tax on the car models that are subject to export restraints. In particular, they assume a constant marginal cost $mc_j$ and thus obtain the first-order condition for the firm producing car model $j$ as:
\begin{align*}
    \ln(mc_j) = \ln(p_j-b_j(p,x,\xi,\theta)-\lambda VER_j)=w_j'\gamma+\omega_j,
\end{align*}
where $b_j$ captures the firm's mark-up (see \Cref{sec:empiricalappendix}), $\lambda$ is the implied tax rate attributed to voluntary export restraints, $\theta = (\bar\beta', \alpha, \sigma)$ represents the demand-side coefficients, and $VER_j$ is a dummy for whether the car is subject to restraint. The $b_j$ mark-up must account for the effect of the price of model $j$ on the rest of the models produced by the same firm. 

We simulate the model in the VER paper with the voluntary export restraint policy present, effectively reproducing the results in the VER paper. Following the VER paper, we then use their estimates to obtain counterfactual market outcomes: the outcomes that would have been obtained in the absence of voluntary export restraints. The compensating variation is calculated taking the observed prices (under the voluntary export restraint policy) as baseline, and the counterfactual prices (without the policy) as the alternative. Table 8 in the VER paper displays an average per-household compensating variation of \$317 ``for those who purchased a car.'' We take the latter qualification to mean those consumers who were likely to purchase a car. In our case, we get an average of \$319 when we consider the consumers who have a probability of choosing the outside option (meaning not buying a car) below the maximum, and
we get an average of \$331 when considering the consumers who have a probability of at least 70\% of purchasing some car. These numbers should be interpreted as welfare losses (the VER paper presents it as a negative number, but according to our convention, they are positive). They are similar to what the VER paper reports in Table 8, and seem to be the largest welfare losses in any of the years analyzed in the VER paper (see their Table 9).\footnote{The BLP can be easily reproduce by following \cite{conlon2020best}. The VER paper, however, presents some challenges that are not present in BLP, and we are not aware of any existing efforts to reproduce their findings. In appendix~\ref{sec:empiricalappendix} we provide additional details on how we reproduced the VER paper.}

Figure~\ref{fig:CV1987} shows the distribution of compensating variation in the population for the counterfactual shutting down of the voluntary export restraint policy. The two graphs correspond to the assumption of non-maximal probability of choosing the outside option and to a probability of buying a car of at least 70\%. The conclusion that emerges from these pictures differs from the story told by the average CV. The median CV is, in particular, zero (when rounding to the nearest integer). In fact, the VER paper remarks that the policy's effects on prices were modest, which is probably why most agents did not face a significant welfare loss as a consequence of voluntary export restraints. This stands in contrast with the conclusion drawn by the VER paper from their Table 8. By looking at the whole distribution of CV, and in particular at the median, we can therefore qualify the conclusion in the VER paper that ``auto purchasers were adversely affected by a significant amount.''

\begin{center}
\begin{figure}
   \includegraphics[width=0.45\textwidth]{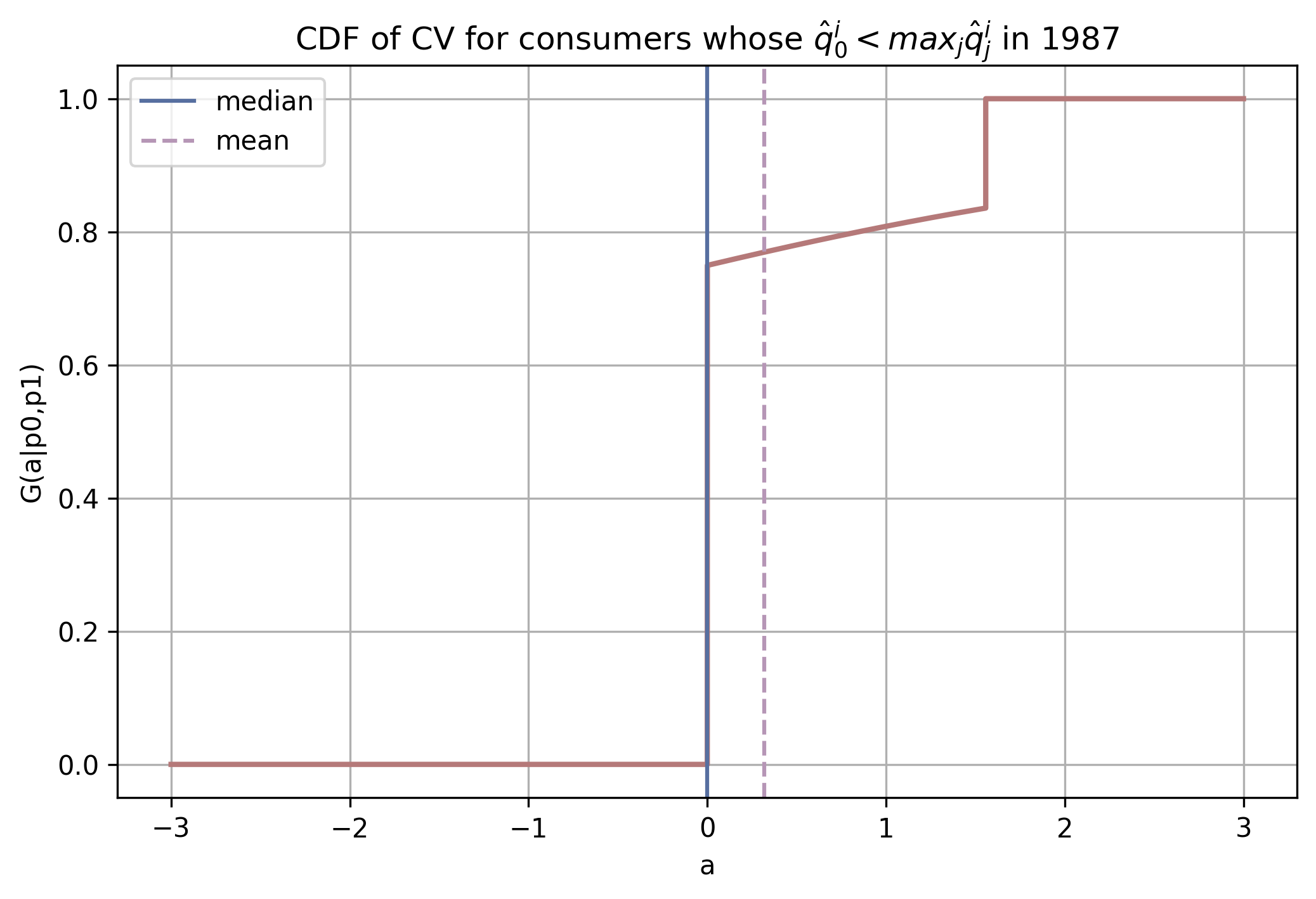}
   \includegraphics[width=0.45\textwidth]{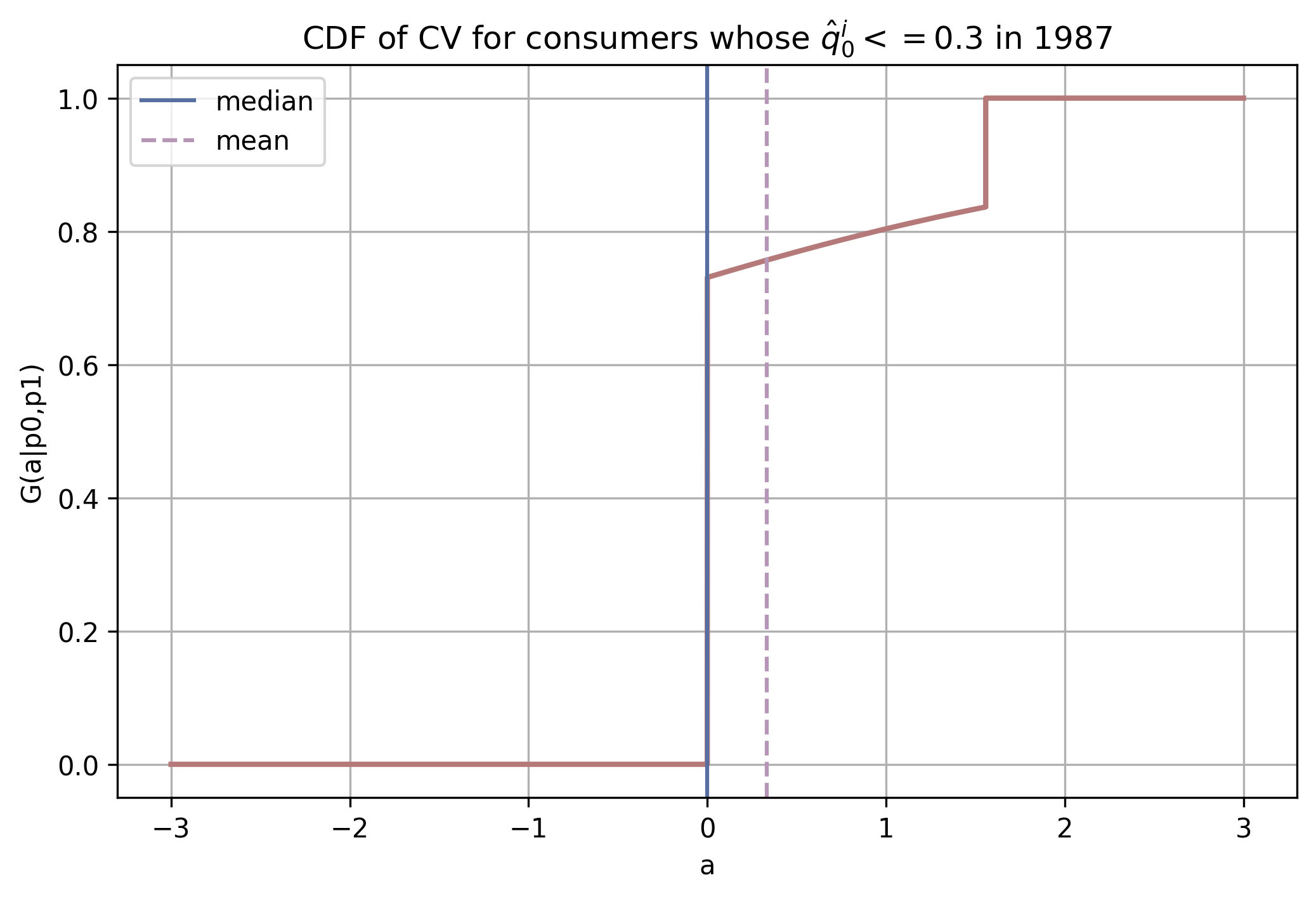}
\caption{CV distribution voluntary export restrictions.}\label{fig:CV1987}
\end{figure}
\end{center}

\section{Why random social choice?}\label{sec:further_discussion}

Besides using social choice as a basis for applied  welfare analysis, we may also be directly interested in the choice function. Our main proposal, weighted utilitarian social choice, involves a planner that makes random choices. We have already emphasized that our main result, \Cref{thm:main}, implies that the planner has to randomize or violate the normative properties in the theorem. In other words, if the planner does not randomize, then they cannot respect Pareto. In addition, in our discussion of applied welfare analysis in \Cref{sec:distributional}, we have argued that a random choice model can better deal with distributional issues than a deterministic criterion. Here we present three additional arguments in favor of a planner that chooses at random.

First, consider the critique of Harsanyi's theorem in \cite{diamondcardinal}. Diamond proposes the following example (slightly adapted to fit with our model). Suppose that there are two agents and two policies, or collective outcomes. Consider the choice between two lotteries, $x$ and $y$: $x$ delivers policy 1 for sure while $y$ delivers each policy with probability $1/2$. The two agents' vNM utilities are $u_1=(1,0)$ and $u_2=(0,1)$. Diamond argues that a utilitarian objective with equal weights for both agents would be indifferent between $x$ and $y$, while a basic fairness consideration should favor lottery $y$. 

Our model does what Diamond argues for. By randomizing between the two agents, and implementing the chosen agent's optimal policy, our model would choose each policy with equal probability. 

Second, we argue that random social choice is a natural solution when agents' individual choices exhibit a Condorcet cycle. Consider an example with $\nagents=3$ agents and $\abs{\prizes}=3$. The details of the example are spelled out below, but the overall idea is simple. Each agent $i$'s random expected utility is concentrated around $u_i\in\Re^3$, with $u_1+u_2+u_3=0$. A small perturbation is added to $u_i$ so that the random utility conform with the regular REU model: two of the agents' perturbations have mean zero, but one of them has a (small) non-zero mean. By a standard normalization, we may regard $u_i$, and any lottery, as a vector in $\Re^2$ (for details, see the proof of \Cref{thm:main}). There are three lotteries, $x,y,z$ that we represent as vectors in $\Re^2$. These lotteries are chosen so that
$u_1\cdot x>u_1\cdot y> u_1\cdot z$,
$u_2\cdot z>u_2\cdot x> u_2\cdot y$, and
$u_3\cdot y>u_3\cdot z> u_3\cdot x$.
The random utilities are concentrated around $u_i$; it turns out, then, that the probability of choosing from each pair of lotteries conforms to the following table:
\[\begin{array}{l|ccc}
	& (x,\{x,y\}) & (y,\{y,z\}) & (z,\{x,z\})  \\ \hline
\rho_1	& \geq 1-\eta & \geq 1-\eta & <\eta \\
\rho_2	& \geq 1-\eta & <\eta	 & \geq 1-\eta \\
\rho_3	& <\eta & \geq 1-\eta & \geq 1-\eta, \\ 
\end{array}
\] with $\eta\in (0,1)$ being small. 

Now the question is: what should be chosen from $\{x,y,z\}$? There is (with high probability) a Condorcet cycle in majority voting over pairs of alternatives. The expected value of the sum of agents' vNM utilities (the planner's utility that would result from applying Equation~\eqref{eq:expU}) is $(-8\da/\pi,0)$, for small $\da>0$; and this expected vNM index will rank the lotteries as $z\succ y\succ x$.

In consequence, a planner that obeys the utilitarian aggregation of  expected utilities will deterministically choose $z$ from $\{x,y,z\}$. With high probability, this choice will coincide with 3's choice, but differ from what 1 or 2 would choose. This seems unfair: the intransitivity in a Condorcet cycle should, arguably, be resolved by randomizing between the favorite choices of each of the individual agents.

\begin{example}
We proceed with the details of the example. 
Suppose that $i$'s vNM utility is the random variable $\tilde u_i = (1-\eta)\bar u_i + \eta \tilde v_i$, where $\bar u_i$ is a deterministic utility and $\tilde v_i$ is random. Think of $\eta\in (0,1)$ as small. Suppose that $u_1=(1,0)$, $u_2=(-1/2,\frac{\sqrt{3}}{2})$ and $u_3=(-1/2,\frac{-\sqrt{3}}{2})$.

Denote by $w(\ta)=(\cos(\ta),\sin(\ta))$ the vector on the unit circle corresponding to angle $\ta$. Suppose that $\tilde v_i$, $i=1,2$, is obtained from $w(\tilde \ta)$, where $\tilde \ta$ is drawn uniformly at random from $[0,\pi/2]$. For $\tilde v_3=w(\tilde \ta)$ suppose that $\tilde \ta$ is drawn from a ``triangular'' density $f$: piecewise linear with $f(\pi)=\frac{1}{2\pi}+\da$ and $f(0)=f(2\pi)=1-\da$ (a linear interpolation in between). Note then that $\sum_i u_i=0$, while $\E \tilde v_i=0$ for $i=1,2$; and $\E\tilde v_3 = (-8\da/\pi,0)$.

Now consider three lotteries, $x,y$ and $z$. Let $\ta_0=\pi/6$, so that $u_i=w(2(i-1)\ta_0)$, $i=1,2,3$. Let $x=w(\ta_0-\ep)$, $y=w(5\ta_0-\ep)$ and $z=w(3\ta_0-\ep)$. Then we see that
\[x\cdot u_1 = \cos (\ta_0-\ep)> y\cdot u_1 = \cos (\ta_0+\ep)>z\cdot u_1 = \cos (3\ta_0-\ep). \] The rankings induced by $u_2$ and $u_3$ can be similarly verified. The expected value $\E\sum_i \tilde u_i$ is equal to $(-8\da/\pi,0)$, and we can verify that $(-8\da/\pi,0)\cdot z>(-8\da/\pi,0)\cdot y> (-8\da/\pi,0)\cdot x$. So $\E\sum_i \tilde u_i$ chooses deterministically $z$ from $\{x,y,z\}$.
\end{example}

Third, the meaning of a social preference is not self-evident, and we argue that a random utility model has some clear advantages over the model that postulates a unique social preference. 

\cite{harsanyi1955cardinal} discusses the possibility that each individual agent has their own social preference, which corresponds to their ethical value judgments over social outcomes: ``each individual is supposed to have a social welfare function of his own, expressing his own individual values --in the same way as each individual has a utility function of his own, expressing his own individual taste.'' Yet, in his model, Harsanyi then insists of the emergence of a single social preference that any individual should hold. Our model allows for a distribution over social preferences in the population, thus respecting the principle that different individuals could hold different ethical value judgments.

Our notion of weighted utilitarianism  can be equivalently understood as random dictatorship \citep{zeckhauser1969majority,zeckhauser1973voting,gibbard1977manipulation}. Standard (deterministic) dictatorships are, of course, undesirable. But random dictatorships are, arguably, sensible solutions to many of the problems we have discussed above. This point is emphasized by \cite{barbera1979majority}. A deterministic dictatorship is unfair because it ignores the preferences of all but one agent. A random dictatorship gives all agents a chance of having their preferences respected. One can also imagine a repeated implementation, in which agents take turns in being selected as the dictator for the group. In the long run, then, all agents get to choose an outcome for the collective.

\section{Alternative interpretation of utilitarian choice rules}\label{sec:alignment}

The axiom of respecting Pareto, which characterizes utilitarian social choice as defined by \Cref{thm:main}, is intuitively related to alignment between the choices of the planner and those of the agents. In this section we provide an alternative interpretation of this condition based directly an a notion of alignment between stochastic choice rules. 

\begin{definition}
    Say that choice rule $\rho'$ is \df{more aligned with the agents} at $D \in \mathcal{D}$ than $\rho$ if for every $i \in \{1,\dots,m\}$, the agreement between $\rho'$ and $\pi_i$ is at least as large as the agreement between $\rho$ and $\pi_i$.
\end{definition}

We make use of the following additional definitions
\begin{itemize}
    \item A choice rule $\rho$ is \df{maximal} at $D$ if there is no other choice rule which is more aligned with the agents at $D$ than $\rho$. 
    \item A choice rule $\rho$ \df{dominates} $\rho'$ if $\rho$ is more aligned with the agents than $\rho'$ at all $D \in \mathcal{D}$.
    \item A choice rule $\rho$ is \df{undominated} if there is no other choice rule that dominates it. A random expected utility $\pi$ is undominated if the induced choice rule $\rho_{\pi}$ is.\footnote{Note that in order to be undominated, an REU must also be undominated by choice rules that are not REU.} 
    \item The planner $\pi$ is a \df{full-support weighted utilitarian} for $\pi_1, \dots, \pi_m$ if there exist weights $\alpha_1, \dots, \alpha_m > 0$, $\sum_{i=1}^{m} \alpha_i = 1$, so that $\pi = \sum_{i = 1}^m \alpha_i \pi_i$.
    \item Given a choice set $D$ and a subset $B \subseteq D$, let $p^u(\cdot |B) \in \Delta(D)$ be the distribution that is uniform over $B$, and zero elsewhere. 
    \item A choice rule $\rho$ is \df{responsive} to $\pi_1, \dots, \pi_m$ if for any $D \in \mathcal{D}$, $x \in D$, and $1 \leq i \leq m$, if $\rho_{\pi_i}(x,D) > 0$ then $\rho(x,D) > 0$.
    \item A choice rule $\rho$ is \df{consistent} with $\pi_1, \dots, \pi_m$ if for any $D \in \mathcal{D}$ and $x \in D$, if $\rho_{\pi_i}(x,D) = 0$ for all $1 \leq i \leq m$ then $\rho(x,D) = 0$.
\end{itemize}

Given a choice set $D$, let $S(D) := \cup_{i=1}^m supp\left(\rho_{\pi_i}(\cdot,D)\right)$ be the union of the supports of each agent's choice rule on $D$.

The following provides an alternative characterization of the social choice rule from \Cref{thm:main}. 

\begin{proposition}\label{prop:undominated}
    The following are equivalent
    \begin{enumerate}
        \item There exists a choice rule that is responsive and undominated.
        \item Every responsive and undominated choice rule is consistent. 
        \item Every consistent choice rule is undominated.
        \item The choice rule defined by $D \mapsto p^u(\cdot| S(D))$ is local behavioral utilitarian for $\pi_1, \dots, \pi_m$.
    \end{enumerate}
\end{proposition}

\begin{proof}
    Proof in \Cref{proof:undominated}. 
\end{proof}

One interpretation of \Cref{prop:undominated} is the following. Suppose we want to use a choice rule that is responsive. We can think of responsiveness as guaranteeing that the planner takes into account the preferences of every agent, at least to a minimal extent. Given this, we would like to avoid choice rules that are dominated, if possible. We might even be willing to use a choice rule that is not weighted utilitarian, or even not REU. \Cref{prop:undominated} tells us that this is not necessary. Either the criteria of full support and undominated cannot be jointly satisfied, or they are satisfied by any choice rule that is responsive and undominated. In particular, any full-support weighted utilitarian choice rule meets these criteria. Intuitively, the condition that $D \mapsto p^u(\cdot| S(D))$ is local behavioral utilitarian for $\pi_1, \dots, \pi_m$ should be satisfied when there are sufficiently many agents, with sufficient disagreement among them.

\section{Literature}\label{sec:literature}

On the applied side, our work is conceptually related to the literature on revealed welfare weights in public finance. The basic idea in this literature is to assume that the social planner's preferences admit a utilitarian representation (the Bergson-Samuelson approach), and then use observed policy choices to recover the implied weights. For example, \cite{christiansen1977theoretical} inverts the optimal taxation problem of \cite{mirrlees1971exploration} to recover welfare weights. This is sometimes referred to as the ``inverse-optimum'' approach; see \cite{hendren2020measuring} for a recent contribution and review. A complementary approach to revealed social preference is taken by \cite{backus2024surplus}. Unlike the inverse-optimum approach, \cite{backus2024surplus} assumes that the planner maximizes consumer surplus (similar to the average compensating variation approach) and then recovers the implied weights on individual welfare using consumer choice data. In line with \Cref{prop:logit_intro} and \Cref{prop:logit}, \cite{backus2024surplus} show that the implied welfare weights are higher for high-income individuals. 

Conversely, \cite{saez2016generalized} study a model in which individual welfare weights are a primitive which define social preferences. They show how individual welfare weights can be used to capture distributional considerations of the planner, and study optimal tax policies under various specifications. This approach is appealing when there is a clear metric by which to make interpersonal comparisons, e.g. money.

Our work differs from the inverse-optimum approach, \cite{backus2024surplus}, and \cite{saez2016generalized} in that we do not assume that social preferences involve linear aggregation of individual's utilities. Our representation of utilitarian social choice, and the implied distributional welfare measures, do not involve interpersonal comparisons of cardinal utility. Moreover, since our distributional welfare measures are based on quantiles, rather than averages, they capture a different set of distributional considerations. 

The distributional welfare criteria which we derive from the characterization of utilitarian social choice relate to the literatures which study quantile-based representations of individual and collective choice. As early as \cite{hanemann1996welfare} the median compensating variation was recognized as a potential criterion for welfare analysis. However, this measure has lacked a conceptual foundation. \cite{manski1988ordinal} and \cite{rostek2010quantile} provide decision-theoretic foundations for a representation of individual choice under uncertainty via quantile maximization. In the context of collective choice, \cite{roberts1980possibility} studies a model with interpersonally-comparable utilities and derives a ``positional dictatorship'' rule under which the social value of an alternative is the welfare of the $q^{th}$-worst-off individual. \cite{bhattacharya2009inferring} adopts this approach to study socially-optimal group formation. \cite{manski2014quantile} study a version of  the \cite{wald1949statistical} treatment assignment problem in which the mean is replaced by the median for evaluating performance.

Most closely related to our work is \cite{chambers2007ordinal}. Like \cite{roberts1980possibility}, the objects of choice in Chambers' model are distributions over individual welfare. \cite{chambers2007ordinal} characterizes the set of statistics on these distributions which satisfy ordinal invariance and monotonicity conditions. These statistics turn out to be (essentially) the quantiles of the welfare distribution. 

Our work differs from this literature in that our characterization is based on quantiles of the compensating variation distribution, which is a measure of welfare changes, as opposed to quantiles of the welfare levels. To understand this distinction, note that under \cite{roberts1980possibility}, the $q^{th}$-worst-off individual under alternative $x$ may differ from that under alternative $y$, whereas our approach would be analogous to looking at the individual whose welfare change when going from $x$ to $y$ is the $q^{th}$-highest. 

Additionally, unlike \cite{chambers2007ordinal} and \cite{rostek2010quantile}, we begin with individuals' choices over alternatives (e.g. price vectors ) and derive an aggregation of these into a social choice function. This function represents the (potentially stochastic) choices of a utilitarian planner. We then show that the planner's (stochastic) compensating variation corresponds exactly to the median of individual compensating variation, under a potentially distorted distribution. Thus rather than starting from an individual welfare measure and deriving a quantile-based representation of social choice, we provide a unified derivation of both the welfare measure (compensating variation) and the quantile rule directly from the social choice function.

Our characterization of utilitarian social choice belongs to the large literature on utilitarian preference aggregation building on \cite{harsanyi1955cardinal}. \cite{hammond1992harsanyi} offers a discussion of Harsanyi's theorem and puts the result in context. \cite{border1985more} provides one of the first rigorous proofs of Harsanyi's theorem, clarifying the role of some assumptions in Harsanyi's original paper that turn out to be superfluous. See also \cite{duggan1996note}. Most proofs use a separation argument that was absent in \cite{harsanyi1955cardinal}, but \cite{zhou1997harsanyi} shows that a very general version of the theorem may be obtained along the lines of Harsanyi's original argument (see also the proof in \cite{camacho1974cardinal}). 

A related issue concerns settings with subjective uncertainty and the simultaneous aggregation of utilities and beliefs: \cite{gilboa2004utilitarian} offer a solution (by weakening the Pareto axiom) to impossibility results discussed by \cite{hylland1979impossibility} and \cite{mongin1995consistent}. Our model has focused on the original setting in Harsanyi which concerns objective lotteries.

We also contribute to the literature on random collective choice. Examples include \cite{fishburn77}, \cite{fishburn1984probabilistic}, and \cite{brandl2016consistent}. \cite{chambers2024correlated} considers correlation in individual random utility, a topic that we ignore altogether. \cite{gibbard1977manipulation} characterized the strategy-proof random social choice mechanisms and obtains essentially a characterization of random dictatorship. See also \cite{barbera1979majority}, who argues that random dictatorships do not share the negative message of deterministic dictatorships as they have some degree of fairness.

{
\singlespacing
\bibliographystyle{ecta}
\bibliography{references}
}

\appendix

\section{Proof of \texorpdfstring{\Cref{thm:main}}{}}\label{proof:main}

\subsection{Preliminary definitions and notation.}
Let $\Re^n_+$ denote the non-negative vectors in $\Re^n$. If $A,C\subseteq \Re^n$ and $\ta\in\Re$ then $\ta A = \{\ta x:x\in A \}$, and $A+C=\{x+y:x\in A, y\in C\}$. We abuse notation and write $A+\{x\}$ as $A+x$. The ball of radius $\ep$ and center $x$ in $\Re^n$ is denoted by $B_\ep(x)$. The unit ball, $B_1(0)$, is simply written as $B$, and its boundary, the unit sphere, as $S$. When $A\subseteq \Re^n$, $A^o$ denotes the interior of $A$.

Let $M(D,u)=\{x\in \cvh (D):u\cdot x\geq u\cdot y \ \ \forall \ y\in \cvh(D) \}$ denote the set of maximizers of the von Neumann-Morgenstern (vNM) index $u$ in $\cvh(D)$, and  $N(D,x)=\{u\in U:u\cdot x\geq u\cdot y \ \ \forall\  y\in \cvh(D)\}$ the set of vNM $u$ for which $x\in M(D,u)$. Let $N^+(D,x)=\{u\in U:u\cdot x > u\cdot y \ \ \forall \ y\in  \cvh(D)\setminus \{x\}\}$. 

A \df{cone} is a subset of $\Re^n$ that is closed under positive scalar multiplication. So $C$ is a cone if $x\in C$ and $\la>0$ implies $\la x\in C$. 
A \df{convex cone} is a subset of $\Re^n$ that is closed under positive linear combinations of its elements. A cone is \df{polyhedral} if it consists of all the positive linear combinations of a finite set. Thus a polyhedral cone is always a convex cone. Observe that $N(D,x)$ is a convex cone, for any $x\in D$, and that $N(D,x)$ is polyhedral when $D\in \D$. A convex cone is \df{pointed} if 0 is one of its extreme points.

For any set $A\subseteq \Re^{n+1}$, denote by $A'\subseteq \Re^n$ the set $\{(x_1,\ldots,x_n):(x_1,\ldots,x_n,x_{n+1})\in A\}$ obtained by projecting $A$ onto its first $n$ coordinates. Observe that if $C$ is a cone and open in $\Re^{n+1}$, then so is $C'$ in $\Re^n$.  

A subset $P$ of $\R^n$ is called a \textit{\textbf{polyhedron}} if there exists an $\ell \times n$ matrix $A$ and a vector $b \in \R^{\ell}$ such that $P = \{x \in \R^n : Ax \leq b \}$. A \textbf{\textit{polytope}} is a subset of $\R^n$ that is the convex hull of a finite set $A \subset \R^n$. So $P$ is a polytope if $P = \cvh(A)$ for some finite set $A \subset \R^n$. Equivalently, $P$ is a polytope if and only if it is a bounded polyhedron \citep{minkowski1896geometrie}.

Let $P$ be a polytope in $\R^n$. If $c$ is a non-zero vector and $\delta = \max \{c\cdot x : x\in P \}$ then the hyperplane $H = \{x \in \R^n : c \cdot x  = \delta\} $ is called a \textit{\textbf{supporting hyperplane}} of $P$. A subset $F$ of $P$ is called a \textit{\textbf{facet}} of $P$ if $F = H \cap P \neq P$ for some supporting hyperplane $H$ of $P$, and there does not exist another supporting hyperplane $H' \neq H$ of $P$ such that $F \subset H' \cap P$. 

Let $\K^*$ be the set of all polyhedral cones in $U$, and $\K$ the set of all pointed polyhedral cones. \cite{gul2006random} show that if $K\in \K^*$ then there exists $D\in \D$ with $K=N(D,0)$ and $0\in D$. Let $\F$ be the smallest field that contains $\K^*$, and $\H=\{\textrm{ri} K : K\in \K \}$.

\subsection{Preliminary lemmas}

\begin{lemma}\label{lem:GPprop4} If $C$ is a polyhedral cone in $U$, then there exists a finite set $D$ and $x\in D$ so that $C=N(D,x)$.
\end{lemma}

\begin{lemma}\label{lem:GPprop6} $\H$ is a semiring and $\F$ consists of finite unions of elements of $\H$.
  \end{lemma}

Lemmas~\ref{lem:GPprop4} and~\ref{lem:GPprop6} are, respectively, Propositions~4 and~6 in \cite{gul2006random}.

\begin{lemma}\label{lem:lowdim} If $C$ is a polyhedral cone in $U$ of dimension less than $\dimU$, and $\hat\pi$ is a regular random expected utility, then $\hat\pi(C)=0$. 
\end{lemma}
\begin{proof}Let $L$ be an affine subspace of dimension less than $\dimU$ that contains $C$. By an additive translation, we may assume that $L$ is a linear subspace (such a translation will not affect $\hat\pi(C)$).

By regarding $C$ as a subset of $L$, and applying Lemma~\ref{lem:GPprop4}, there exists a finite $D\subseteq L$ and $x\in D$ with $C=N(D,x)\cap L$. Choose $z\in X$ in the orthogonal complement of $L$, and consider $D'=\{y+z:y\in D\}\cup D$. Such $z$ exists because the dimension of $L$ is less than $\dimU$. Then for $y\in D$ and $u \in C$, $u\cdot(y+z)=u\cdot y$, as $u$ is in $L$ and therefore orthogonal to $z$. This means that any vNM index in $C$ is tied between $x$ and $x+z$ in $D'$. Hence $\hat\pi(C)=0$ because $\hat\pi$ is a regular random expected utility.
  \end{proof}

\begin{lemma}\label{lem:Nint}If $D\in \D$, and $x$ is an extreme point of $D$, then $N(D,x)$ has non-empty interior. 
\end{lemma}
\begin{proof} First, assume that $\cvh(D)$ has dimension $n$. Let $F_1,\ldots,F_k$ be the facets of $\cvh(D)$ that $x$ belongs to. Let $u_i\in U$ and $\ta_i$ define the facet $F_i$, so $\cvh(D)\subseteq \{y\in\Re^{n+1} : u_i\cdot y\leq \ta_i\}$ while $F_i=\cvh(D)\cap \{y\in\Re^{n+1} : u_i\cdot y= \ta_i\}$ for all $i$. Fix $\la_j>0$, $1\leq j\leq k$.

For any $y\in \cvh(D)$, $y\neq x$, there is some facet $F_i$ that $y$ does not belong to; thus $u_i\cdot y< u_i\cdot x=\ta_i$. Hence, $y\cdot \sum_{j=1}^k \la_j u_j<x\cdot \sum_{j=1}^k \la_j u_j$. So for any $\la_j>0$, $\sum_{j=1}^k \la_j u_j\in N^{+}(D,x)$, and therefore $N^{+}(D,x)$ is open in $U$.

If $\cvh(D)$ has dimension less than $n$, then let $L$ be an affine subspace of dimension less than $n$ that contains $C$. By additive translation, we may assume that $L$ is a linear subspace. Then by the argument above, we conclude that $N^+(D,x)$ is open in $L$. Moreover, for any $u \in U$ in the orthogonal complement of $L$, any $z \in N^+(D,x) \cap L$, and any $x \neq y \in \cvh(D)$, we have $(u + z) \cdot y < (u + z) \cdot x$. Thus $u + z \in N^+(D,x)$, and so $N^+(D,x)$ is open in $\mathcal{D}$. 
\end{proof}

\subsection{A separation argument}

We first prove the result under the assumption that $\pi$ is regular. It is then easy to show that regularity is necessary to respect Pareto, which we do in \Cref{sec:regular}.

Let $\ba(U,\F)$ be the space of \df{charges} (finitely additive signed measures) on $(U,\F)$ endowed with the total variation norm -- a Banach space. The topological dual of $\ba(U,\F)$ is well-known to be intractable. We follow instead the approach of \cite{rao1983theory} in working with a related space. Let $\mu = \frac{1}{m+1}(\pi+\sum_{i=1}^m\pi_i)$. Then $\hat\pi\ll \mu$ (denoting absolute continuity) for all $\hat\pi=\pi,\pi_1,\ldots,\pi_m$. Denote by $\mathbf{P}$ the set of all $\F$-measurable partitions of $U$, and let
\[ 
\norm{\la} \coloneqq \lim_{\{F_1,\ldots,F_J \}\in\mathbf{P}}\sum_{j=1}^J \abs{\la(F_j)},
\] where the limit is taken over the net defined by partition refinement: see Proposition 7.1.2 in \cite{rao1983theory}, who show that this is well defined (partitions are finite). 

Let $V(U,\F,\mu) \coloneqq \{\la\in \ba(U,\F): \la\ll\mu \text{ and } \norm{\la}<\infty \}.$ Observe that, for each $\hat\pi\in\{\pi,\pi_1,\ldots,\pi_m\}$, $\hat \pi \ll\mu$ and $\sum_{j=1}^J \hat\pi(F_j)$ is bounded by $\hat\pi(U)$ for all $\{F_1,\ldots,F_J \}\in\mathbf{P}$. Thus we may regard  each of the countably additive measures $\pi$, $\pi_1,\ldots,\pi_n$ as an element of $V(U,\F,\mu)$.

By Theorem 7.2.3 in \cite{rao1983theory}, $(V(U,\F,\mu),\norm{\cdot})$ is a Banach space. By the proof of Theorem 9.2.3 in \cite{rao1983theory}, for any continuous linear functional $T$ on $V(U,\F,\mu)$ there is a bounded charge $\la\in\ba(U,\F)$ with 
\[ 
T(\nu) = \lim_{\{F_1,\ldots,F_J \}\in\mathbf{P}}\sum_{j=1}^J \nu(F_j)\frac{\la(F_j)}{\mu(F_j)},
\] where again the limit is obtained with the partitions in $\mathbf{P}$ directed by refinement. 

Suppose that $\pi$ is not a weighted utilitarian aggregation of $\pi_1,\ldots,\pi_n$. So $\pi$ is not an element of the convex hull of $\{\pi_1,\ldots,\pi_n\}$ in $V(U,\F,\mu)$. Since the space $V(U,\F,\mu)$ is Banach, \cite{schaeffer} Theorem II.10.2 implies that the convex hull, $\cvh (\{\pi_1,\ldots,\pi_n\})$, is a compact set. Then, by a strict version of the separating hyperplane theorem (\cite{schaeffer} Theorem II.9.2), the exists a continuous linear functional $T:V(U,\F,\mu)\to\Re$ strictly separating $\pi$ from $\cvh (\{\pi_1,\ldots,\pi_n\})$. Thus $T(\pi-\pi_i)>0$ for all $i=1,\ldots,m$. Let $\la$ be the bounded charge associated with $T$.

Now, because $i$ ranges over a finite set, by taking a fine enough partition $\{F_1,\ldots,F_J \}\in\mathbf{P}$, we obtain that 
\[
\sum_{j=1}^J [\pi(F_j)-\pi_i(F_j)]\frac{\la(F_j)}{\mu(F_j)} >0.
\] for all $i$. By Lemma~\ref{lem:GPprop6}, each measurable set $F_j\in \F$ is the union of elements of $\H$, and for any $H\in \H$ with $H\subseteq F_j$, the semiring property implies that $F_j\setminus H$ is a finite union of disjoint elements of $\H$. So by considering a further partition refinement we may without loss of generality assume that each $F_j\in \H$.

Note that $\la$ may take negative values. So define $\la'(F_j) =\la(F_j)+\gamma \mu(F_j)$ when $\mu(F_j)>0$ and $\la'(F_j)=1$ when $\mu(F_j)=0$, and choose $\gamma>0$ so that $\la'(F_j)>0$. Note that $\mu(F_j)>0$ for at least one $F_j$, and that when $\mu(F_j)=0$ we know that $\pi(F_j)=\pi_i(F_j)=0$. Using the convention that $0=0/0$ (which underlies the analysis in \cite{rao1983theory}) we have that, for any~$i$,
\begin{align*}
    \sum_{j=1}^J [\pi(F_j)-\pi_i(F_j)]\frac{\la'(F_j)}{\mu(F_j)} & = 
    \sum_{j:\mu(F_j)>0} [\pi(F_j)-\pi_i(F_j)]\frac{\la'(F_j)}{\mu(F_j)} \\
    & = \sum_{j=1}^J [\pi(F_j)-\pi_i(F_j)]\frac{\la(F_j)}{\mu(F_j)} + \gamma\sum_{j=1}^J [\pi(F_j)-\pi_i(F_j)] \\
    & = \sum_{j=1}^J [\pi(F_j)-\pi_i(F_j)]\frac{\la(F_j)}{\mu(F_j)}>0,
\end{align*} as $\sum_{j=1}^J [\pi(F_j)-\pi_i(F_j)]=0$ because the $F_j$ are a partition of $U$, and $\pi$ and $\pi_i$ are probability measures.

Now if we first define  $c_j=\frac{\la'(F_j)}{\mu(F_j)}>0$ for $\mu(F_j)>0$, and $c_j=0$ when $\mu(F_j)=0$ then we have $c_j\geq 0$ and $\sum_j c_j>0$. So we may normalize these numbers and obtain that $\sum_j c_j=1$ while $\sum_{j=1}^J [\pi(F_j)-\pi_i(F_j)]c_j >0$ for all $i=1,\ldots,m$.

Finally by Lemma~\ref{lem:lowdim}, we may take each of the cones $F_j$ to be of dimension $\dimU$ at the expense that the collection $F_j$ may no longer exhaust $U$, while remaining pairwise disjoint.

\subsection{Connecting the separation to the axiom}\label{sec:sphere_approx}

Recall that all the vectors $u$ in $U$ have $u_{\dimU+1}=0$, so we may identify a cone $C$ in $U$ with $C'$ in $\Re^\dimU$. In consequence, if $\hat \pi$ is one of the measures $\pi$, $\pi_1,\ldots, \pi_m$ we denote by $\hat \pi'$ the unique measure obtained by setting $\hat\pi'(C')=\hat\pi(C)$ for all measurable $C$ in $U$.

Let $B\coloneqq\{x\in\Re^n:\norm{x}\leq 1\}$ be the unit ball, and $S\coloneqq\{x\in\Re^n:\norm{x}=1 \}$ the unit sphere, in $\Re^n$. We will show that for any finite subset, $\{x'_1, \dots, x'_J\}$, of the unit sphere in $S \subset \Re^n$, there is a decision problem $D = \{x_1, \dots, x_J\} \subset X$ such that $\rho_{\hat{\pi}}(x_j,D) = \hat{\pi}'(\{u \in \Re^n : u \cdot x_j \geq u \cdot x_l \ \forall \ 1 \leq l \leq J \})$. We first construct a suitable subset of $S$, and then prove this claim, which allows us to move back into the spaces $X$ and $U$ which are subsets of $\Re^{n+1}$.

Choose finite sets $D_k\subseteq S$, so that $D_k \subset D_{k+1}$ and $\cup_{k=1}^\infty D_k$ is dense in $S$. Then $\cvh(D_k)\subseteq B$, and $\cvh(D_k)\to B$ in the Hausdorff metric. We claim that, for  $k$ large enough, 
\begin{equation}\label{eq:convDk}
\sum_{j=1}^J c_j \pi'(\{u : M(D_k,u)\cap F'_j\neq \os\}) > \sum_{j=1}^J c_j \pi'_i(\{u : M(D_k,u)\cap F'_j\neq \os\}) \quad \forall \ i.    
\end{equation}

To this end, fix $j\in \{1,\ldots,J\}$ and define $f^k: B \rightarrow \Re$ by 
\begin{equation*}
  f^k(u) =\begin{cases}
  1 & \text{ if } M(D_k,u)\cap F'_j\neq \os \\
  0 & \text{ if } M(D_k,u)\cap F'_j= \os \\
  \end{cases}
\end{equation*}

\textit{Claim:} $f^k$ converges pointwise to $f \coloneqq \one_{\{F'_j\}}$ $\hat\pi$-almost everywhere, for any $\hat\pi\in\{\pi,\pi_1,\ldots,\pi_m\}$. 

For any $u\in B$, $M(B,u)=\{u\}$. Let $x_k\in M(D_k,u)$ be arbitrary, and note that $M(D_k,u)\subseteq B$ for all $k$, and $B$ is compact, so $\{x_k\}$ has at least one limit point. Now $\cvh(D_k)\to B$ and the maximum theorem imply that $x_k\to u$ for any convergent subsequence of $\{x_k\}$; hence we conclude that $x_k\to u$.

The cone $F'_j$ is open, so for any $u\in F'_j\cap B$, there is $k$ large enough so that $M(D_k,u)\cap F'_j$ is nonempty. Conversely, suppose that $u\in B$ is not in the closure of $F'_j$.  If it were the case that  $M(D_k,u)\cap F'_j\neq\os$ infinitely often, then we could choose $x_k\in M(D_k,u)$ that has a limit point in the closure of $F'_j$, which would contradict $x_k\to u$. So $M(D_k,u)\cap F'_j=\os$, for $k$ large enough, whenever $u$ is not in the closure of $F'_j$. By regularity, $\hat\pi'(F'_j) = \hat\pi'(\bar F'_j)$. So $f^k\to  \one_{\{F'_j\}}$ $\hat\pi$-a.e. We may now prove~\eqref{eq:convDk}. Since $f_k$ is bounded, by the Bounded Convergence Theorem we have
\begin{align*}
  \hat \pi'(\{u : M(D_k,u)\cap F'_j\neq \os\}) & = \int f_k \diff \hat\pi'  
  \to \int f\diff \hat\pi' \\
   & = \hat \pi' (\{u : M(B,u)\cap F'_j\neq \os \} \\
   & = \hat \pi'(F'_j)
\end{align*}
for any $\hat{\pi}' \in \{\pi',\pi'_1, \dots,\pi'_m\}$. There are finitely many agents and sets $F_j$, so we have uniform convergence across $\pi'$, $\pi'_i$, and $F_j$.

We have that $\sum_{j=1}^J c_j \pi'(F'_j) > \sum_{j=1}^J c_j \pi'_i(F'_j)$ for $i=1,\ldots,m$ from our previous argument. Hence we may fix $k$ large enough so that
\[\begin{split}
\sum_{j=1}^J c_j \pi'(\{u : M(D_k,u)\cap F'_j\neq \os\}) 
> 
\sum_{j=1}^J c_j \pi'_i(\{u :  M(D_k,u)\cap F'_j\neq \os\}) 
\end{split}\] for $i=1,\ldots,m$.

Finally, we bring the construction to $X$ and $U$, which are subsets of $\Re^{\dimU+1}$. Note that if $\ta>0$ is a scalar and $z\in\Re^n$, then $x\in M(D_k,u)$ if and only if $\ta x + z\in  M(\ta D_k+z,u)$, as linear preferences are both homothetic and translation invariant. Hence, 
\[
\hat \pi'(\{u : M(\ta D_k+z,u)\cap (\ta F'_j + z) \neq \os\}) 
= \hat \pi'(\{u : M(D_k,u)\cap F'_j \neq \os\}).
\] Moreover $\ta F'_j=F'_j$ as $F'_j$ is a cone. 

Now we may choose $\ta>0$ and $z$ so that $\ta D_k + z\subseteq \{x\in\Re^n_+:\sum_{i=1}^n x_i\leq 1 \}$. Define $C_k\subseteq \Re^{n+1}$ by \[
C_k = \{(x_1,\ldots,x_n,1-\sum_{i=1}^n x_i): x\in \ta D_k + z \}.
\] Note that $C'_k=\ta D_k+z$ and $C_k\subseteq X=\Delta(\prizes,2^\prizes)$.

Recall that $u_{n+1}=0$ for all $u\in U$, so we may identify $u\in \Re^n$ with $(u,0)\in U$. Then $x\in M(\ta D_k+z,u)$ iff $(x_1,\ldots,x_n,1-\sum_{i=1}^n x_i) \in M(C_k,(u,0))$.

To sum up, then, we have
\[\begin{split}
\sum_{j=1}^J c_j \pi(\{u\in U : M(C_k,u)\cap F_j\neq \os\})  >
\sum_{j=1}^J c_j \pi_i(\{u\in U : M(C_k,u)\cap F_j\neq \os\}) 
\end{split}\]

Finally, note that
\begin{align*}
\hat \pi(\{u\in U : M(C_k,u)\cap F_j\neq \os\}) 
& = \sum_{\{x\in C_k:x\in F_j\}} \hat \pi(\{u\in U: M(C_k,u)=x \}) \\
& =\rho_{\hat\pi}(F_j\cap C_k | C_k).
  \end{align*}
for any $\hat\pi=\pi,\pi_1,\ldots,\pi_m$. We conclude that $\sum_{j=1}^J c_j \rho_{\pi}(F_j\cap C_k | C_k)>
\sum_{j=1}^J c_j \rho_{\pi_i}(F_j\cap C_k | C_k)$ for $i=1,\ldots,m$.

We prove that $\pi$ respects Pareto iff for every decision problem $D\in\D$, $(\rho_\pi(x,D))_{x\in D}$ is a convex combination of $A=\{ (\rho_{\pi_1}(x,D))_{x\in D},\ldots,(\rho_{\pi_m}(x,D))_{x\in D} \}$. Indeed one direction of this statement is obvious. So suppose that $(\rho_\pi(x,D))_{x\in D}$ is not in $\cvh (A)$. Since the latter set is compact, by a version of the separating hyperplane theorem, there exists $c\in\Re^D$, $c\neq 0$, so that $\sum_{x\in D} c_x\rho_{\pi}(x,D)>\sum_{x\in D} c_x\rho_{\pi_i}(x,D)$ for $i=1,\ldots,m$. Thus $\pi$ does not respect Pareto.

\subsection{Necessity of regularity}\label{sec:regular}
We now show that if $\pi$ is not regular then it does not respect Pareto. In \Cref{sec:sphere_approx} above we have shown the following: each $\hat{\pi} \in \{\pi, \pi_1,\dots,\pi_m\}$ is identified with a unique measure $\hat{\pi}'$ on $\Re^n$, and for any finite subset, $\{x'_1, \dots, x'_J\}$, of the unit sphere in $S \subset \Re^n$, there is a decision problem $D = \{x_1, \dots, x_J\} \subset X$ such that $\rho_{\hat{\pi}}(x_j,D) = \hat{\pi}'(\{u \in \Re^n : u \cdot x_j \geq u \cdot x_l \ \forall \ 1 \leq l \leq J \})$. We therefore work directly in $S$ and $\Re^n$ here, and do not repeat the translation back to the spaces $X$ and $U$. 

If $\pi$ is not regular, then it must have at least one atom $u_0\in U$, which we may take to have unit norm. Thus no matter what tie-breaking rule is used to define choice from $D \subset S$ in case of indifference, if $u_0 \in D$ is an extreme point of $D$, then $\rho_{\pi}(u_0, D) \geq \theta>0$ for some $\theta$. This holds since $u_0 \cdot u_0 > u_0 \cdot x$ for all $u_0 \neq x \in S$. Consider an enumeration of a countable dense set $u_0,u_1,\ldots$ in $S$. And let the decision problem $D_k$ equal $D_k=\{u_0,\ldots,u_k \}$. Then, for any $i$, $\pi_i(N(D_k,u_0))\to 0$, as $\cap_k N(D_k,u_0) = \{ \la u_0:\la >0\}$ and $\pi_i$ is regular. Choose then $k$ with $\pi_i(N(D_k,u_0))<\theta \leq \rho_{\pi}(u_0, D) \text{ for all } i=1,\ldots, m.$ Define $c:D_k\to \Re$ to be $c(x)=0$ for $x\neq u_0$ and $c(u_0)=1$. Then $\E_{\rho_{\pi}(.,D_n)}c = \theta>\E_{\rho_{\pi_i}(.,D_n)}c$ for all $i$. Hence $\rho_{\pi}$ does not respect Pareto. 

\section{Other omitted proofs}

\subsection{Proof of \texorpdfstring{\Cref{lem:median_negative}}{}} \label{proof:median_negative}
\begin{proof}
    As a counterexample, suppose there are two possible preference types, $\varepsilon_1$ and $\varepsilon_2$, which are equally likely in the population. Assume that all consumers have no wealth effects and are risk neutral with respect to monetary lotteries (i.e. their utilities are quasilinear). Let $x$ be a binary 50-50 lottery supported on $\{p^1,p^2\}$ such that 
    \begin{equation*}
    \begin{split}
        \CV(p^0,p^1|\varepsilon_1) = 2 \quad & \quad \CV(p^0,p^1|\varepsilon_2) = -1 \\
        \CV(p^0,p^2|\varepsilon_1) = -1 \quad & \quad \CV(p^0,p^2|\varepsilon_2) = 2.
    \end{split}    
    \end{equation*} 
    Then all consumers prefer $x$ to $p^0$ since $\frac{1}{2} 2 - \frac{1}{2}1 > 0$, so by the Pareto property we should have $x \succsim_{med} p^0$. However $\CV_{med}(p^1) = CV_{med}(p^2) = -1$, so under weak consistency we must have $p^0 \succ_{med} x$.
\end{proof}

\subsection{Proof of \texorpdfstring{\Cref{prop:foundation}}{}}\label{proof:foundation}

\begin{proof}
    Given the random utility $\hat{\pi}_{\alpha}$, the choice probability $\rho^{\sigma}\big((p^0,y),\{(p^0,y), (p',y-x) \}\big)$ can be written as
    \begin{equation*}
      \sum_{i \in I} \sigma_i \int \one\{ v(p^0,y) \geq v(p',y-x|\varepsilon) \} \hat{F}_i(d\varepsilon) = \int \one\{v(p^0,y) \geq v(p',y-x|\varepsilon) \} \hat{F}(d\varepsilon) := d(x). 
    \end{equation*}
    where the equality follows from the definition of $\hat{F}$. Then by the definition of compensating variation, if $x > \CV_{med}(p^0,p')$ then there exists $x' < x$ such that $d(x') \geq \frac{1}{2}$. Similarly, if $x < \CV_{med}(p^0,p')$ then $d(x) < \frac{1}{2}$.
\end{proof}

\subsection{Proof of \texorpdfstring{\Cref{prop:CV_cdf}}{}}\label{proof:CV_cdf}

\begin{proof}
    By \cite{bhattacharya2018empirical} Theorem 1, for a single type, $i$, the CDF of -CV is given by
\begin{equation*}
    F_i(a|p^0,p') =
    \begin{cases}
        0, \quad &\text{if } a < \Delta p_1 \\
        \sum_{x = 1}^j &\hat{q}^i_x(p'_1,\dots ,p'_j, p^0_{j+1} + a, p^0_{j+2} + a, \dots, p^0_{Z} + a, y+a)\\
          & \text{if } \Delta p_j \leq a < \Delta p_{j+1}, a \geq 0, 1 \leq j \leq Z-1\\
        \sum_{x = 1}^j &\hat{q}^i_x(p'_1 - a,\dots ,p'_j - a, p^0_{j+1}, p^0_{j+2}, \dots, p^0_{Z}, y)\\
          & \text{if } \Delta p_j \leq a < \Delta p_{j+1}, a < 0, 1 \leq j \leq Z-1\\
        1, &\text{if } a \geq \Delta p_Z.
    \end{cases}
\end{equation*}
(Bhattacharya defines compensating variation so that it is negative for price decreases). By summing across types, we can replace $\hat{q}^i_x$ with $Q^{\alpha}_x$ in the above expression to yield the CDF of the negative of compensating variation in the population. To derive the expression for $G_{\hat{\pi}_\alpha}$ we use the fact that $\Pr(-x \leq a) = 1- \Pr(x < -a) = 1 - \lim_{b \nearrow -a}\Pr(x \leq -b)$. 
\end{proof}

\subsection{Proof of \texorpdfstring{\Cref{prop:single_price_median}}{}}\label{proof:single_price_median}

\begin{proof}
     Consider the case of a price decrease for good 1, so compensating variation is supported on $[0,-\Delta p_1]$. Suppose the CDF $G_{\hat{\pi}_{\alpha}}$ is convex, as depicted in \Cref{fig:med_proof}. For a non-negative random variable the mean is the area $B + C + D$ (this can be seen by integrating by parts). The median is $A + B + C$. By drawing a subgradient to the CDF at the median it is easy to see that $A > D$ must hold if there is no mass point at zero. In this case the median is above the mean. This continues to hold as long as the mass at zero is sufficiently low which, by \Cref{cor:single_price_cdf}, occurs when the share of consumers who purchase good 1 at prices $p'$ is sufficiently high. The other cases are symmetric. 
\end{proof}

    \begin{figure}[ht]
        \centering
    \begin{tikzpicture}[scale = 0.8]
        \draw[thick] (0,0) rectangle (6,4);
    \node at (-0.2,4) {$1$};
    \node at (-0.2,0) {$0$};
    \node at (6.2,0) {$-\Delta p_1$};
    
        \draw[thick, domain=0:6, smooth, variable=\x] plot ({\x}, {(\x^2)/9});

        \draw[thick, dashed] (4.24, 2) -- (0,2) node[anchor = east] {$\frac{1}{2}$};
    \draw[thick, dashed] (4.24, 4) -- (4.2,0) node[anchor = north] {median};

        \node at (3,.5) {A};
    \node at (2,1) {B};
    \node at (2,3) {C};
    \node at (5,3.5) {D};
    
\end{tikzpicture}
    \caption{Median above mean for convex CDF}
    \label{fig:med_proof}
    \end{figure}

\subsection{Proof of \texorpdfstring{\Cref{prop:logit}}{}}\label{proof:logit}

\begin{proof}
    From \Cref{prop:single_price_median}, we just need to show that $a \mapsto \hat{q}_Z^i(p_{-Z}', p^0_{Z} + a, y + a)$ is convex. Under the logit specification
    \begin{equation*}
        \hat{q}_Z^i(p_{-Z}', p^0_{Z} + a, y + a) = \dfrac{\exp\big(W^i_Z(y - p^0_Z)\big)}{\exp\big(W^i_Z(y - p^0_Z)\big) + \sum_{z=1}^{Z-1} \exp\big(W^i_z(y + a - p'_z)\big)}.
    \end{equation*}
    Note that $a \mapsto \frac{1}{f(a)}$ is decreasing and convex if $f$ is increasing and concave. The denominator in the above expression is concave in $a$ under the assumption that $\exp(W_x^i(\cdot))$ is concave for all $i$ and $x$.  
\end{proof}

\subsection{Proof of \texorpdfstring{\Cref{prop:undominated}}{}}\label{proof:undominated}

We first establish an intermediate result characterizing when a choice rule is maximal at some $D \in \mathcal{D}$. 

    \begin{lemma}\label{lem:maximal}
    $\rho$ is maximal at $D$ iff there exists $y \in \Delta(\{1,\dots,m\})$ such that $supp(\rho(\cdot,D)) \subset \argmax\{ y^TA_D [x] : x \in D \}$
    \end{lemma}
    \begin{proof}
        At a choice set $D$, $\rho$ is maximial iff $\min_{p' \in \Delta(D)} \max_{y \in \Delta(\{1,\dots,m\})} y^T A_D (\rho(\cdot,D) - p') \geq 0,$ where $A_D$ is a matrix with the $i^{th}$ row equal to the vector $\rho_{\pi_i}(\cdot,D)$. By the minimax theorem, this is equivalent to $\max_{y \in \Delta(\{1,\dots,m\})} \min_{p' \in \Delta(D)} y^T A_D (\rho(\cdot,D) - p') \geq 0.$ In other words, there must exist a $y \in \Delta(\{1,\dots,m\})$ such that $\max_{p' \in \Delta(D)} y^T A_D p' \leq y^T A_D \rho(\cdot,D).$ The product $y^T A_D$ is an element of $\cvh(\{\rho_{\pi_1}, \dots, \rho_{\pi_m}\})$. The maximization problem on the left hand side of the previous equation is solved by setting $p'(x) = 1$ for any $x \in \argmax\{ y^TA_D [x] : x \in D \}$. Thus for $\rho$ to be maximal, it must be supported only on $ \argmax\{ y^TA_D [x] : x \in D \}$. 
\end{proof}

\noindent\textit{Proof of \Cref{prop:undominated}.}
    $(1) \Rightarrow (2)$. Let $\rho$ be a responsive and undominated choice rule. Fix an arbitrary choice set $D$. Since $\rho$ is responsive, $supp( \rho(\cdot,D)) \supseteq S(D)$. Then by \Cref{lem:maximal}, there must exist $y \in \Delta(\{1,\dots,m\})$ such that $supp( \rho(\cdot,D)) \subseteq \argmax \{ y^T A_D [x] : x \in D \}$. Since $y \in \Delta(\{1,\dots,m\})$, we have $\argmax \{ y^T A_D [x] : x \in D \} \subset S(D)$, so it must be that $\argmax \{ y^T A_D [x] : x \in D \} = S(D) = supp(\rho(\cdot,D))$, which implies that $\rho$ is consistent.

    $(2) \Rightarrow (4)$. From above, we have that for each $D$ there exists a $y$ such that $\argmax \{ y^T A_D [x] : x \in D \} = S(D)$. In other words, $y^T A_D [x]$ must be uniform on $S(D)$. The existence of such a $y$ for each $D$ is exactly the definition of $D \mapsto p^u(\cdot| S(D))$ being local behavioral utilitarian. 

    $(4) \Rightarrow (3)$. For each $D$, let $y_D$ be the weights such that $\sum_{i=1}^m y_D(i) \rho_{\pi_i}(\cdot,D) = p^u(\cdot,S(D))$. Then $\argmax\{y_D^T A_D[x] : x \in D \} = S(D)$ by construction. Since $supp(\rho(\cdot,D)) \subseteq S(D)$ for any consistent $\rho$, the condition of \Cref{lem:maximal} is satisfied. Since this holds for all $D$, any consistent $\rho$ is undominated. 

    $(3) \Rightarrow (1)$. This follows from the fact that there always exist responsive and consistent (and REU) choice rules: any full-support weighted utilitarian planner does the trick. \hfill\qed

\newpage
\section{Online Appendix: Omitted Details of Empirical Application}
\label{sec:empiricalappendix}

Our discussion in Section~\ref{sec:empirical} focuses on our results for the VER paper, \cite{berry1999voluntary}. In applying our proposals to the BLP methodology, we use the estimates from the BLP paper to simulate counterfactual related to VER restrictions, as described below. The \texttt{pyblp} package of \cite{conlon2020best} simplifies this task substantially and contains the data from BLP. In Figure~\ref{fig:CV1990}, we report one set of results about general price changes in the BLP model: unrelated to the voluntary export restraint policy. We bring this up as another illustration of our results. 

Using the data for 1990, we obtained the CV distribution for an artificial 10\% increase in the price of all cars, compared to a 10\% increase in the price of ``luxury cars.'' The price increase for all cars results in a median CV of  \$765 with a mean of \$746. In contrast, the increase in the price of luxury cars results in a median CV of  \$0 with a mean of \$168 (we have rounded here to the nearest integers as there are some instabilities inherent in our simulations). So again we find a distributional effect that would seem to exaggerate the welfare impact of a price increase when we focus on the mean CV. The choice of the year 1990 is arbitrary. We obtain similar results for other years.
\begin{figure}[h]
\centering
   \includegraphics[width=0.45\textwidth]{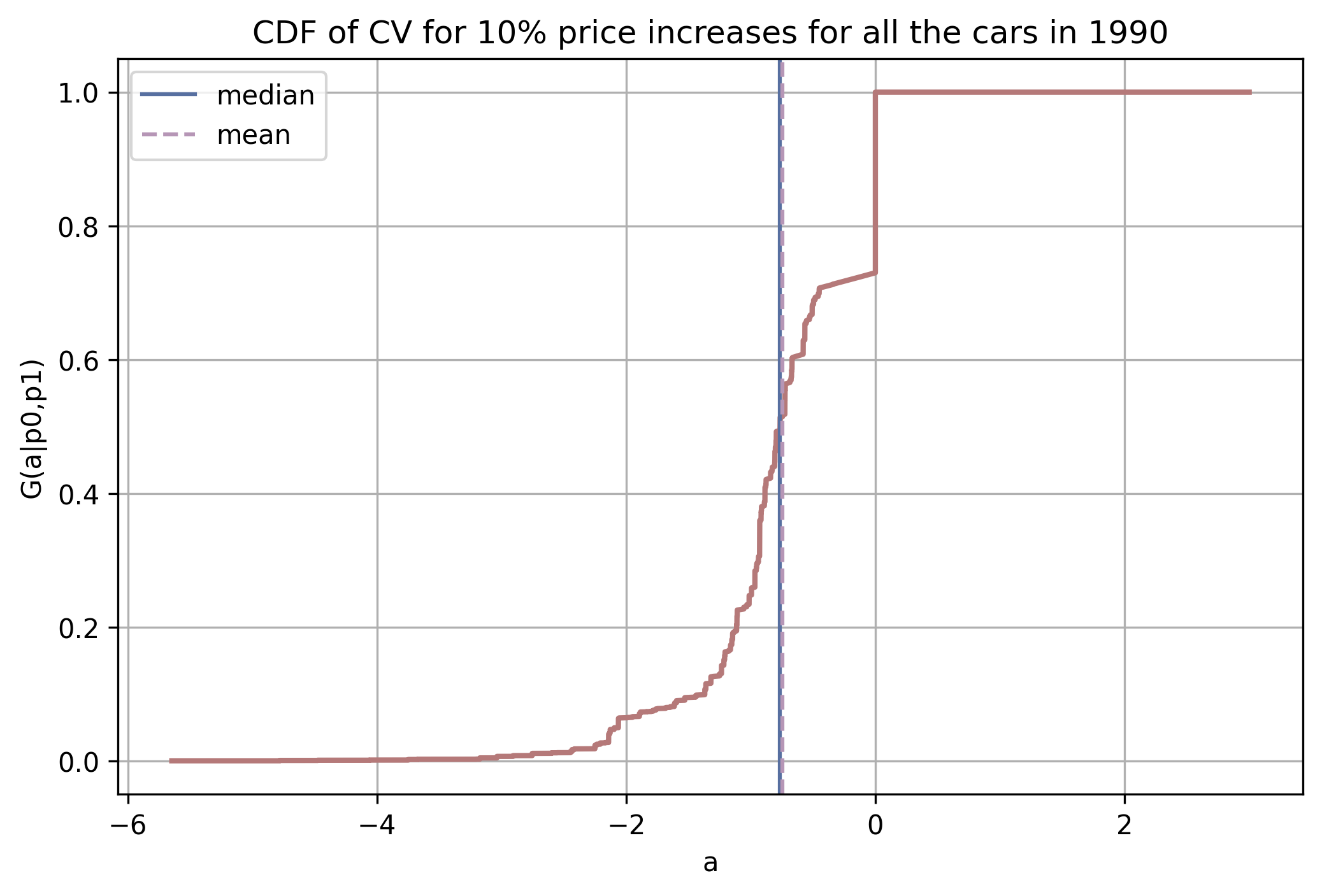}
   \includegraphics[width=0.45\textwidth]{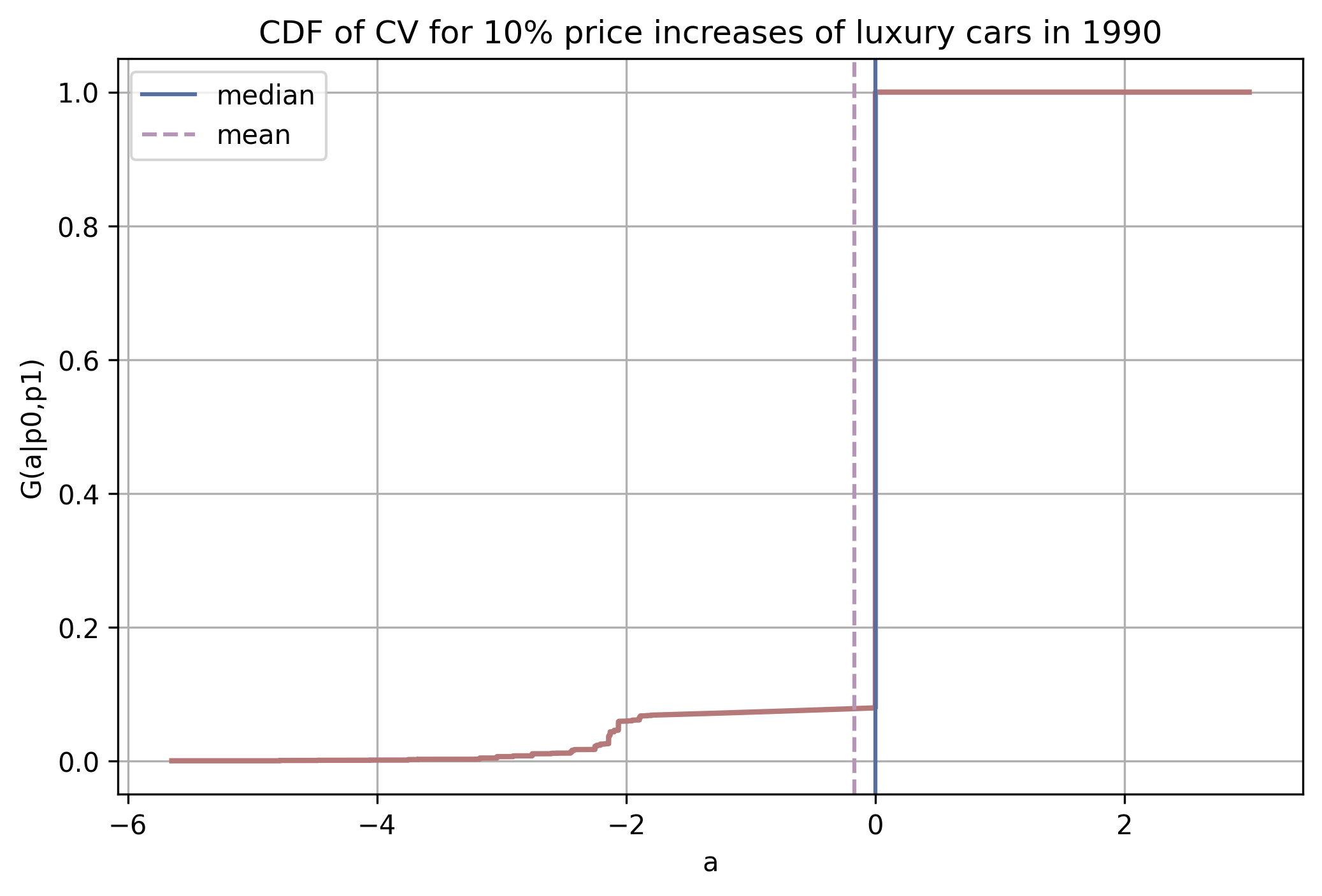}
\caption{CV distribution luxury car price increases vs.\ all cars.}\label{fig:CV1990}
\end{figure}

The rest of this appendix provides some details on how we obtained the numbers reported in Section~\ref{sec:empirical}. As said, reproducing the BLP paper is greatly simplified by the methodology of \cite{conlon2020best}. For the VER paper, the situation is more complicated and we could not just apply the \texttt{pyblp} package ``off the shelf.'' With the assistance of our RA, Jinglin Yang, we developed the following approach. 

Our general strategy is to take the estimates from the BLP paper for the distribution of preferences and use these to calculate the welfare effects of the price changes induced by the VER restrictions. First, note that the demand side of the model, as described in the body of our paper, coincides with the original BLP model with the addition of a non-linear dependence on income. The supply side is different and accounts for the presence of the VER policy, but the \texttt{pyblp} package is not compatible with their specification of marginal costs. This necessitates a different method for recovering $\xi_j$ and $\omega_j$ from $\delta_j$. 

The $\delta_j$ parameters are recovered from the demand-side of the model. If $s_i$ are market shares, we obtain 
    \begin{equation*}
        s_j=\sum_{i\in I} w_i \hat q_j^i(p,y)=\sum_{i\in I}w_i \frac{\exp(\delta_j+\mu_{ij})}{1+\sum_{j=1}^J\exp(\delta_j+\mu_{ij})}:=s_j(\delta),\text{ where }\delta=(\delta_1,\ldots,\delta_J)',\label{equ: share}
    \end{equation*} 
with $w_i$ being the weight of agents $i$ used for importance sampling (which may not sum to 1). We obtain the weights used in BLP (1995) from \texttt{pyblp} package.

We take the estimated values of $\alpha,\sigma_k$, and $v$ from the paper, from which it is easy to construct $\mu_{ij},\forall i, j$. Then we back out $\delta_j$ using a simple algorithm \citep{berry1995automobile}: 
\begin{enumerate}
        \item Choose starting value $\delta^0=\ln(s)-\ln(s_{0})$, where $s=(s_1,...,s_J)$
        \item Calculate the next value $\delta^n = \delta^{n-1}+\ln(s)-\ln\left(s(\delta^{n-1})\right)$, where $s$ are observed market shares and $s(\delta^{n-1})$ are the market shares implied by the model, which are given by equation \ref{equ: share}.
        \item Iterate until $||\delta^n-\delta^{n-1}||_{\infty}=\max\{|\delta_{jct}^n-\delta_{jct}^{n-1}|\}<10^{-14}$\vspace*{10pt}
    \end{enumerate}
With the $\delta_j$'s in hand, we back up the mark-up function $b_j(p,x,\xi;\theta):=b_j(p)$. The  firm's maximization problem is $\max_{{p_j:j\in J_f}}\sum_{j\in J_f}(p_j-mc_j-\lambda VER_j)M s_j(p)-C_f$. So the FOC implies that 
    \begin{align*}
        &\hspace{40pt}s_j(p)+\sum_{r\in J_f}(p_r-mc_r-\lambda VER_r)\frac{\partial s_r(p)}{\partial p_j}=0,\forall j\in J_f\\
        &{\implies \underbrace{s(p)}_{J\times 1}-\underbrace{\Omega}_{J\times J} \underbrace{(p-mc-VER\cdot \lambda )}_{J\times 1}=0\implies p-mc- VER\cdot \lambda= \Omega^{-1}s(p)
        }
    \end{align*}
    where $\Omega=S\otimes O$. 
    $S$ denotes the $J\times J$ matrix of demand derivatives whose $(j,k)$ entry is given by $S_{jk}=-\partial s_k/\partial p_j$, and $O$ denotes the $J\times J$ ownership matrix whose $(j,k)$ entry is given by $O_{jk}=\mathds{1}(\exists f,j,k\in J_f)$.

It turns out that \[S_{jk}=-\frac{\partial s_k(p)}{\partial p_j}=-\sum_i w_i \frac{\partial \hat q_k^i(p)}{\partial p_j}=\alpha \sum_i w_i\frac{\hat q_k^i(p)\left[
            \mathds{1}(j=k)-\hat q_j^i(p)
        \right]}{y_i}\]
Therefore, 
    \begin{align*}
        \Omega_{jk}(p)=S_{jk}\times O_{jk}=\alpha \sum_i w_i\frac{\hat q_k^i(p)\left[
            \mathds{1}(j=k)-\hat q_j^i(p)
        \right]}{y_i}\times \mathds{1}(\exists f,j,k\in J_f)
    \end{align*}
    which implies that 
    \begin{align*}
        b_j(p)=\sum_{k=1}^J \Omega^{-1}_{jk}(p)s_k(p)=\sum_{k=1}^J \Omega^{-1}_{jk}(p)\sum_i w_i \hat q_k^i(p),\forall j = 1,2...,J
    \end{align*}

In sum, with estimates for $\delta_j$'s, $\alpha, \sigma$ and with data on weights $w_i$'s and incomes $y_i$'s, we construct the function $q_k^i(p),\forall i,\forall k=1,2,\ldots,J$, which records the probability of consumer $i$ choosing goods $k$ given prices $p$. Then, combining with information on the ownership matrix, we  construct $\Omega_{jk}(p)$ and thus $b_j(p)$ for each good. 

Now, given the marginal cost specification $\ln(mc_j)=\ln(p_j-b_j(p)-\lambda VER_j)=w_j'\gamma+\omega_j,$ with the parameter estimates on $\lambda,\gamma$ from the paper and data on the VER dummies $VER_j$, we back out the unobservable factors in the marginal cost specification, $\omega_j = \ln(p_j-b_j(p)-\lambda VER_j)-w_j'\gamma$. Given these estimates for $\omega_j$'s, to determine the counterfactual prices without the VER restriction, we set $\lambda$ (the implicit tax) to zero, and solve for the vector of prices $p^*$ which satisfies the following equations, $\ln(p_j^*-b_j(p^*))=w_j'\gamma+\omega_j,\forall j= 1,2,..,J$. This gives us the price changes implied by the VER restrictions, to which we then apply the results of \Cref{sec:distributional}.

\end{document}